\documentclass[global,
]{svjour}

\usepackage[latin1]{inputenc}
\usepackage{latexsym}
\usepackage{graphics}
\usepackage{times}
\usepackage{mathrsfs,amsmath,amsbsy,amssymb,amsfonts,txfonts,dsfont}

\newcommand{\pg}[1]{\left\{#1\right\}}

\newcommand{\pt}[1]{\left(#1\right)}
\newcommand{\vir}[1]{``#1''}
\renewcommand{\d}{\mathrm{d}}
\def\1{\mathbf{1}}

\newcommand{\Data}{X^{(n)}} 
\newcommand{\DataXe}{X^{(n)}} 

\begin{document}
\title{
Bayes and maximum likelihood for 
$L^1$-Wasserstein deconvolution of Laplace mixtures
}


\author{}


\author{Catia Scricciolo
\thanks{Catia Scricciolo\\ 
Dipartimento di Scienze Economiche, Universit\`a degli Studi di Verona,
Polo Universitario Santa Marta, Via Cantarane 24,
I-37129 Verona (VR), ITALY, \email{catia.scricciolo@univr.it}}}

\institute{}

%

%
\date{
}
%

\titlerunning{Bayes and ML for $L^1$-Wasserstein deconvolution of Laplace mixtures}

\authorrunning{C. Scricciolo}

\maketitle

\begin{abstract}
We consider the problem of recovering a distribution function on the real line from observations additively contaminated with errors following the standard Laplace distribution. Assuming that the latent distribution is completely unknown leads to a nonparametric deconvolution problem. We begin by studying the rates of convergence relative to the $L^2$-norm and the Hellinger metric for the direct problem of estimating the sampling density, which is a mixture of Laplace densities with a possibly unbounded set of locations: the rate of convergence for the Bayes' density estimator corresponding to a Dirichlet process prior over the space of all mixing distributions on the real line matches, up to a logarithmic factor, with the $n^{-3/8}\log^{1/8}n$ rate for the maximum likelihood estimator. Then, appealing to an inversion inequality translating the $L^2$-norm and the Hellinger distance between general kernel mixtures, with a kernel density having polynomially decaying Fourier transform, into any $L^p$-Wasserstein distance, $p\geq1$, between the corresponding mixing distributions, provided their Laplace transforms are finite in some neighborhood of zero, we derive the rates of convergence in the $L^1$-Wasserstein metric for the Bayes' and maximum likelihood estimators of the mixing distribution. Merging in the $L^1$-Wasserstein distance between Bayes and maximum likelihood follows as a by-product, along with an assessment on the stochastic order of the discrepancy between the two estimation procedures.
\end{abstract}




\section{Introduction}
The problem of recovering a distribution function from observations additively contaminated with measurement errors is the object of study in this note. 
Assuming data are sampled from a convolution kernel mixture, the interest is in \vir{estimating} the mixing or latent distribution from contaminated observations.
The statement of the problem is as follows. Let $X$ be a random variable (r.v.) with probability measure $P_0$ 
on the Borel-measurable space $(\mathbb{R},\,\mathscr{B}(\mathbb{R}))$, with Lebesgue density $p_0:=\d P_0/\d \lambda$. 
Suppose that
\[X=Y+Z,\]
where $Y$ and $Z$ are independent,
unobservable random variables,
$Z$ having Lebesgue density $f$. We examine the case where the error 
has the standard Laplace distribution with density
$$f(z)=\frac{1}{2}e^{-|z|}, \quad
z\in\mathbb{R}.$$
The r.v. $Y$ has unknown distribution $G_0$ on some measurable space
$(\mathscr Y,\,\mathscr B(\mathscr Y))$, with $\mathscr Y\subseteq \mathbb{R}$ and 
$\mathscr B(\mathscr{Y})$ the Borel $\sigma$-field on 
$\mathscr Y$.
The density $p_0$ is then
the convolution of $G_0$ and $f$,
\[p_0(x)=(G_0\ast f)(x)=\int_{\mathscr Y}f(x-y)\,\d G_0(y),\quad x\in\mathbb{R}.\]
In what follows, we also write $p_0\equiv p_{G_0}$ to stress the
dependence of $p_0$ on $G_0$. Letting $\mathscr G$ be the set of all probability measures $G$ on 
$(\mathscr Y,\,\mathscr B(\mathscr Y))$, the parameter space
\[\mathscr P:=\Bigg\{p_G(\cdot):=\int_{\mathscr Y}f(\cdot-y)\,\d G(y),\, G\in\mathscr G\Bigg\}\]
is the collection of all convolution Laplace mixtures 
and the model is nonparametric.

Suppose we observe $n$ independent copies $X_1,\,\ldots,\,X_n$ of $X$. 
The r.v.'s $X_1,\,\ldots,\,X_n$ are independent and identically distributed (i.i.d.) according to the density $p_0\equiv p_{G_0}$ on the real line. 
The interest is in recovering the mixing distribution $G_0\in\mathscr G$ from indirect observations.
Deconvolution problems may arise in a wide variety of contexts,
the error distribution being typically modelled as a Gaussian, 
even if also the Laplace has relevant applications. 
Full density deconvolution, together with the related many normal 
means problem, has drawn attention in the literature since the late 1950's
and different deconvolution methods have been proposed and developed since then taking the 
frequentist approach, the most popular being based on nonparametric maximum likelihood and kernel methods. Rates of convergence have been mostly investigated for \emph{density} deconvolution: 
Fan (1991a, 1991b) showed that deconvolution kernel density estimators achieve global optimal rates for weighted $L^p$-risks, $p\geq1$, when the smoothness of the density to be recovered is measured in terms of the number of its 
derivatives. Hall and Lahiri (2008) considered estimation of the \emph{distribution function} using 
the cumulative distribution function corresponding to the deconvolution kernel density estimator
and showed that it attains minimax-optimal pointwise and global rates for the integrated mean-squared error over different functional classes for the error and latent distributions, smoothness being described through the tail behaviour of their Fourier transforms.
For a comprehensive account on the topic, the reader may refer to the monograph of Meister (2009). In this note, we do not assume that the probability measure $G_0$ possesses Lebesgue density. Wasserstein metrics are then particularly well-suited as global loss functions: convergence in $L^p$-Wasserstein
metrics for discrete mixing distributions has, in fact, a natural interpretation in terms of convergence
of the single supporting atoms of the probability measures involved. 
Dedecker \emph{et al}. (2015) have obtained a lower bound on the rate of convergence for the $L^p$-Wasserstein risk, $p\geq1$, when no smoothness assumption,
except for a moment condition, is imposed on the latent distribution and the error distribution is ordinary smooth, the Laplace being a special case.

Deconvolution problems have only recently begun to be studied from a Bayesian perspective: the typical scheme considers the mixing distribution as a draw  
from a Dirichlet process prior. 
Posterior contraction rates for recovering the mixing distribution 
in $L^p$-Wasserstein metrics have been
investigated in Nguyen (2013) and Gao and van der Vaart (2016), even though
the upper bounds in these articles do not match with the lower bound in
Dedecker \emph{et al}. (2015). 
Minimax-optimal adaptive recovery rates for mixing densities belonging to Sobolev spaces 
have been instead obtained by Donnet \emph{et al}. (2018)
in a fully Bayes as well as in an empirical Bayes approach to inference, the latter accounting for a data-driven choice of the prior hyperparameters of the Dirichlet process baseline measure.

In this note, we study nonparametric Bayes and maximum likelihood estimation of the mixing distribution $G_0$, when no smoothness assumption is imposed on it. 
The analysis begins with the estimation of the sampling density $p_0$: estimating the \emph{mixed} density $p_0$ can, in effect, be the first step for recovering the \emph{mixing} distribution $G_0$. 
Taking a Bayesian approach, if the random density $p_G$ is modelled as 
a Dirichlet-Laplace mixture, then $p_0$ can be consistently estimated at a rate 
$n^{-3/8}$, up to a $(\log n)$-factor, if $G_0$ has tails matching with those of the baseline measure of the Dirichlet process, which essentially requires $G_0$ to be in the weak support of the process, see Proposition \ref{prop:2} and Proposition \ref{prop:1}.
This requirement allows to extend to a possibly unbounded set of locations the results of Gao and van der Vaart (2016), which take into account only the case of compactly supported mixing distributions.
Taking a frequentist approach, $p_0$ can be estimated by the maximum likelihood still 
at a rate $n^{-3/8}$, up to a logarithmic factor. 
As far as we are aware, the result on the rate of convergence in the Hellinger metric for the maximum likelihood estimator (MLE) of a Laplace convolution mixture is new and is obtained taking the approach proposed by Van de Geer (1996), according to which it is the \vir{dimension} of the class of kernels and the behaviour of $p_0$ near zero that determine the rate of convergence for the MLE. As previously mentioned, results on the estimation of $p_0$ are interesting in view of 
the fact that, appealing to an inversion inequality translating the Hellinger or the $L^2$-distance between kernel mixtures, with Fourier transform of the kernel density having polynomially decaying tails, into any $L^p$-Wasserstein distance, $p\geq1$, between the corresponding mixing distributions, rates of convergence in the $L^1$-Wasserstein metric for the MLE and the Bayes' estimator of the mixing distribution can be assessed. 
Merging in the $L^1$-Wasserstein metric between Bayes and maximum likelihood for deconvolving Laplace mixtures follows as a by-product.

\medskip


\noindent \emph{Organization}.
The note is organized as follows. Convergence rates in the Hellinger metric for Bayes and maximum likelihood density estimation of Laplace convolution mixtures are preliminarily studied in Sect. \ref {sec:Bayes} and in Sect. \ref {sec:MLE}, respectively, in view of their subsequent instrumental use for assessing the $L^1$-Wasserstein accuracy of the two estimation procedures in recovering the mixing distribution of the sampling density. Merging between Bayes and maximum likelihood follows, as shown in Sect. \ref {sec:merging}. Remarks and suggestions for possible refinements and extensions of the exposed results are presented in Sect. \ref{sec:finrmks}. Auxiliary lemmas, along with the proofs of the main results, are deferred to Appendices A--D.

\medskip

\noindent \emph{Notation}. We fix the notation and recall some 
definitions used throughout.\\[7pt]
{\textsf{Calculus}}
\begin{itemize}
\item[--] The symbols \vir{$\lesssim$} and
\vir{$\gtrsim$} indicate inequalities valid up to a constant multiple that is universal or fixed within the context, but anyway
inessential for our purposes.
\item[--] For sequences of real numbers $(a_n)_{n\in\mathbb{N}}$ and $(b_n)_{n\in\mathbb{N}}$, the notation $a_n\sim b_n$ means that $(a_n/b_n)\rightarrow 1$ as $n\rightarrow+\infty$.
Analogously, for real-valued functions $f$ and $g$, 
the notation $f\sim g$ means that $f/g\rightarrow1$ 
in an asymptotic regime that is 
clear
from the context. 
\end{itemize}
{\textsf{Covering and entropy numbers}} 
\begin{itemize}
\item[--] Let $(T,\,d)$ be a (subset of a) semi-metric space. For every $\varepsilon>0$, the $\varepsilon$-\emph{covering number} of $(T,\,d)$, denoted by $N(\varepsilon,\,T,\,d)$, is defined as the minimum number of $d$-balls of radius $\varepsilon$ needed to cover $T$. 
Take $N(\varepsilon,\,T,\,d)=+\infty$ if no finite covering by $d$-balls of radius $\varepsilon$ exists. The logarithm of the $\varepsilon$-covering number, $\log N(\varepsilon,\,T,\,d)$,
is called the $\varepsilon$-\emph{entropy}.
\smallskip
\item[--] Let $(T,\,d)$ be a (subset of a) semi-metric space. For every $\varepsilon>0$, the $\varepsilon$-\emph{packing number} of $(T,\,d)$, denoted by $D(\varepsilon,\,T,\,d)$, is defined as the maximum number of points in $T$ such that the distance between each pair is at least $\varepsilon$. Take $D(\varepsilon,\,T,\,d)=+\infty$ if no such finite $\varepsilon$-packing exists. The logarithm of the $\varepsilon$-packing number, $\log D(\varepsilon,\,T,\,d)$, 
is called the $\varepsilon$-\emph{entropy}. 
\vspace*{-0.05cm}
\end{itemize}
Covering and packing numbers are related by the inequalities
$$N(\varepsilon,\,T,\,d)\leq D(\varepsilon,\,T,\,d)\leq N(\varepsilon/2,\,T,\,d).$$
{\textsf{Function spaces and probability 
}}
\begin{itemize}\item[--] For real number $1\leq p <+\infty$, let 
$$L^p(\mathbb{R}):=\{f|\,f:\mathbb{R}\rightarrow\mathbb{C},\,f \mbox{ is Borel measurable, }\int|f|^p\,\d\lambda<+\infty\}.$$
For $f\in L^p(\mathbb{R})$, the $L^p$-norm of $f$ is defined as 
$||f||_p:=(\int|f|^p\,\d\lambda)^{1/p}$. The supremum norm of a function $f$ is defined as $||f||_\infty:=\sup_{x\in\mathbb{R}}|f(x)|$.
\item[--] For $f\in L^1(\mathbb{R})$, the complex-valued function $\hat f(t):=\int_{-\infty}^{+\infty} e^{itx}f(x)\,\d x$, $t\in\mathbb{R}$, is called the \emph{Fourier transform of $f$.} 
\item[--] All probability density functions are meant to be with respect to Lebesgue measure $\lambda$
on $\mathbb{R}$ or on some subset thereof. 
\item[--] The same symbol, $G$ (say), is used to denote a probability measure on a Borel-measurable space $(\mathscr Y,\,\mathscr B(\mathscr Y))$
and the corresponding cumulative distribution function (c.d.f.).
\item[--] The degenerate probability distribution putting mass one at a point $\theta\in\mathbb{R}$ is denoted by $\delta_\theta$.
\item[--] The notation $Pf$ abbreviates the expected value $\int f\,\d P$, where the integral is understood to extend over the entire natural domain when, here and elsewhere, the domain of integration is omitted.
\item[--] Given a r.v. 
$Y$ with distribution $G$, the \emph{moment generating function} of $Y$ or the \emph{Laplace transform of the probability measure $G$} is defined as
$$M_G(s):=E[e^{sY}]=\int_{\mathscr{Y}}e^{sy}\,\d G(y) \,\,\,\mbox{ for all $s$ for which the integral is finite.}$$
\end{itemize}
{\textsf{Metrics and divergences}} 
\begin{itemize}
\item[--] The \emph{Hellinger distance} between any pair of probability density functions $q_1$ and $q_2$ on $\mathbb{R}$ is defined as $h(q_1,\,q_2):=\{\int(q_1^{1/2}-q_2^{1/2})^2\,\d \lambda\}^{1/2}$, the $L^2$-distance between the square-root densities. The following inequalities, due to LeCam (1973), p. 40, relating the $L^1$-norm and the Hellinger distance hold:
\begin{equation}\label{eq:Hel_L1}
h^2(q_1,\,q_2)\leq||q_1-q_2||_1
\end{equation} 
and 
\begin{equation}\label{eq:L1_Hel}
||q_1-q_2||_1\leq 2 h(q_1,\,q_2).
\end{equation}
\item[--] For ease of notation, the same symbol $d$ is used throughout to denote the $L^1$-norm, the $L^2$-norm or the Hellinger metric, the intended meaning being declared at each occurrence.
\item[--] For any probability measure $Q$ on $(\mathbb{R},\,\mathscr{B}({\mathbb{R}}))$ with density $q$, 
let 
\begin{align*}\textrm{KL}(P_0\|Q):=
\left\{
\begin{array}{ll}
\displaystyle\int \log\frac{\d P_0}{\d Q}\,\d P_0=\int_{p_0q>0} p_0\log\frac{p_0}{q}\,\d\lambda, & \mbox{\quad if $P_0\ll Q$,}\\[10pt]
\quad +\infty, &\mbox{\quad otherwise,}
\end{array}\right.
\end{align*}
be the \emph{Kullback-Leibler divergence} of $Q$ from $P_0$ and, for $k\geq2$, let
\begin{align*}\textrm{V}_k(P_0\|Q):=
\left\{
\begin{array}{ll}
\displaystyle
\int \bigg|\log\frac{\d P_0}{\d Q}\bigg|^k\,\d P_0=
\int_{p_0q>0} p_0\bigg|\log\frac{p_0}{q}\bigg|^k\,\d \lambda, & \mbox{\quad if $P_0\ll Q$,}\\[10pt]
\quad +\infty, &\mbox{\quad otherwise,}
\end{array}\right.
\end{align*}
be the $k$th absolute moment of $\log(\d P_0/\d Q)$. 
For any $\varepsilon>0$ and a given $k\geq2$, define a Kullback-Leibler type neighborhood of $P_0$ as 
$$B_{\mathrm{KL}}(P_0;\,\varepsilon^k):=\{Q:\,\textrm{KL}(P_0\|Q)\leq\varepsilon^2,\,\textrm{V}_k(P_0\|Q)\leq\varepsilon^k\}.$$
\item[--] For any real number $p\geq 1$ and any pair of probability measures $G_1,\,G_2\in\mathscr G$ with finite $p$th absolute moments,
the $L^p$-\emph{Wasserstein distance} between $G_1$ and $G_2$ is defined as
\[W_p(G_1,\,G_2):=\pt{\inf_{\gamma\in\Gamma(G_1,\,G_2)}\int_{\mathscr Y\times \mathscr Y}|y_1-y_2|^p\,
 \gamma(\d y_1,\,\d y_2)}^{1/p},\]
where $\Gamma(G_1,\,G_2)$ is the set of all joint probability measures on $(\mathscr Y\times \mathscr Y)\subseteq\mathbb{R}^2$,
with marginals $G_1$ and $G_2$ on the first and second arguments, respectively. 
\end{itemize}
{\textsf{Stochastic order symbols}}\\[4pt]
Let $(Z_n)_{n\in\mathbb{N}}$ be a sequence of real-valued random variables, possibly defined on entirely different probability spaces $(\Omega_n,\,\mathscr F_n,\,\mathbf{P}_n)_{n\in\mathbb{N}}$. Suppressing $n$ in $\mathbf{P}$ causes no confusion if it is understood that $\mathbf{P}$ refers to whatever probability space $Z_n$ is defined on. Let $(k_n)_{n\in\mathbb{N}}$ be a sequence of positive real numbers. We write
\begin{itemize}
\item $Z_n=O_{\mathbf{P}}(k_n)$ if
$\lim_{T\rightarrow+\infty}\limsup_{n\rightarrow+\infty}\mathbf{P}(|Z_n|>Tk_n)=0$. Then,
$Z_n/k_n=O_{\mathbf{P}}(1)$,
\item $Z_n=o_{\mathbf{P}}(k_n)$ if, for every $\varepsilon>0$, $\lim_{n\rightarrow +\infty}\mathbf{P}(|Z_n|>\varepsilon k_n)=0$. Then, $Z_n/k_n=o_{\mathbf{P}}(1)$.
\end{itemize}
Unless otherwise specified, in all stochastic order symbols used throughout, the probability measure $\mathbf{P}$ is understood to be $P_0^n$, the joint law of the first $n$ coordinate projections of the infinite product probability measure $P_0^{\mathbb{N}}$.


\section{Rates of convergence for $L^1$-Wasserstein deconvolution of Dirichlet-Laplace mixtures}\label{sec:Bayes}
In this section, we present some results on the Bayesian recovery of a distribution function from data contaminated with an additive random error following the standard Laplace distribution: we derive rates of convergence 
for the $L^1$-Wasserstein deconvolution of Dirichlet-Laplace mixture densities. The density is modeled as a Dirichlet-Laplace mixture
$$p_{G}(\cdot)\equiv (G \ast f)(\cdot)=\int_{\mathscr Y} f(\cdot-y)\,\d G(y),$$ 
with the kernel density $f$ being the standard Laplace  
and the mixing distribution 
$G$ being any probability measure on $(\mathscr Y,\,\mathscr{B}(\mathscr{Y}))$, with 
$\mathscr Y\subseteq \mathbb{R}$.
As a prior for $G$, we consider
a Dirichlet process with base measure $\alpha$ on 
$(\mathscr{Y},\,\mathscr{B}(\mathscr{Y}))$, denoted by $\mathscr{D}_{\alpha}$.
We recall that a Dirichlet process
on a measurable space $(\mathscr{Y},\,\mathscr{B}(\mathscr Y))$,
with finite and positive base measure $\alpha$ on $(\mathscr{Y},\,\mathscr{B}(\mathscr Y))$,
is a random probability measure
$\tilde G$ on $(\mathscr{Y},\,\mathscr{B}(\mathscr Y))$ such that, for every finite partition
$(B_1,\,\ldots,\,B_k)$ of $\mathscr{Y}$, $k\geq1$, the vector of random
probabilities $(\tilde G(B_1),\,\ldots,\,\tilde G(B_k))$ has Dirichlet
distribution with
parameters $(\alpha(B_1),\,\ldots,\,\alpha(B_k))$. A Dirichlet process mixture of Laplace densities can be structurally described as follows:
\vspace{-0.2cm}
\begin{description}
\item[\,\,\,$\bullet$]
\,\,\, $\tilde G\sim\mathscr{D}_{\alpha}$, \vspace*{-0.15cm}
\item[\,\,\,$\bullet$]
\,\,\,  given $\tilde G=G$, the r.v.'s $Y_1,\,\ldots,\,Y_n$ are i.i.d. according to $G$, \vspace*{-0.15cm}
\item[\,\,\,$\bullet$]
\,\,\,  given $(G,\,Y_1,\,\ldots,\,Y_n)$, the
r.v.'s $Z_1,\,\ldots,\,Z_n$ are i.i.d. according to $f$, \vspace*{-0.15cm}
\item[\,\,\,$\bullet$]
\,\,\,
 sampled values from $p_G$ are defined as
$X_i:=Y_i+Z_i$ for $i=1,\,\ldots,\,n$.
\end{description}

Let the sampling density $p_0$ be itself a Laplace mixture with mixing distribution $G_0$, that is, $p_0\equiv p_{G_0}=G_0\ast f$. In order to assess the rate of convergence in the $L^1$-Wasserstein metric for the Bayes' estimator of the true mixing distribution $G_0$, we appeal to an inversion inequality relating the $L^2$-norm or the Hellinger distance between Laplace 
mixed densities 
to any $L^p$-Wasserstein distance, $p\geq1$, between the corresponding mixing distributions, see Lemma \ref{lem:2} in Appendix D. 
Therefore, we first derive rates of contraction in the $L^2$-norm and the Hellinger metric for the posterior distribution of a Dirichlet-Laplace mixture prior: convergence of the posterior distribution at
a rate $\varepsilon_n$, in fact, implies the existence of Bayes' point estimators that converge at least as fast as $\varepsilon_n$ in the frequentist sense. The same indirect approach has been 
taken by Gao and van der Vaart (2016), who deal with the case of compactly supported mixing distributions, while we extend the results to mixing distributions possibly supported on the whole real line or on some unbounded subset thereof. We present two results on posterior contraction rates for a Dirichlet-Laplace mixture prior. The first one, as stated in Proposition \ref{prop:2}, is relative to the $L^1$-norm or the Hellinger metric; the second one, 
as stated in Proposition \ref{prop:1}, is relative to the $L^2$-metric. Proofs are deferred to Appendix C.

\begin{proposition}\label{prop:2}
Let $X_1,\,\ldots,\,X_n$ be i.i.d. observations from a density $p_0\equiv p_{G_0}=G_0\ast f$,
with the kernel density $f$ being the standard Laplace and the mixing distribution $G_0$ such that, for some decreasing function $A_0:\,(0,\,+\infty)\rightarrow [0,\,1]$ and a constant $0<c_0<+\infty$,
\begin{equation}\label{eq:tailG0SS}
G_0([-T,\,T]^c)\leq A_0(T)\lesssim \exp{(-c_0T)}\quad\mbox{for large $T>0$}.
\end{equation}
If the baseline measure $\alpha$ of the Dirichlet process is symmetric around zero
and possesses density $\alpha'$ such that, for some constants $0<b<+\infty$ and $0<\tau\leq 1$,
\begin{equation}\label{eq:tailG1}
\alpha'(y)\propto \exp{(-b|y|^\tau)},\quad y\in\mathbb{R},
\end{equation}
then there exists a sufficiently large constant $M>0$ such that
$$\Pi(d(p_G,\,p_0)\geq M n^{-3/8}\log^{5/8}n
\mid \Data)=o_{\mathbf{P}}(1),$$
where $\Pi(\cdot\mid \Data)$ denotes the posterior distribution corresponding 
to a Dirichlet-Laplace process mixture prior after $n$ observations
and $d$ can be either the Hellinger or the $L^1$-metric.
\end{proposition}

\begin{remark}
In virtue of the following inequality,
$$\forall\,G,\,G'\in\mathscr G,\,\,\, ||p_G-p_{G'}||_2^2\leq 4||f||_\infty h^2(p_G,\,p_{G'}),$$
where $||f||_\infty=1/2$ for the standard Laplace kernel density, 
see \eqref{eq:hel^2} in Lemma \ref{lem:l2hel},
the $L^2$-metric posterior contraction rate for a Dirichlet-Laplace mixture prior could, in principle, be
derived from Proposition \ref{prop:2}, which relies on Theorem 2.1 of Ghosal \emph{et al}.
(2000), p. 503, or Theorem 2.1 of Ghosal and van der Vaart (2001), p. 1239,
but this would impose slightly stronger conditions on 
the density $\alpha'$ of the baseline measure than those required in Proposition \ref{prop:1} below,
which is based on Theorem 3 of Gin\'{e} and Nickl (2011), p. 2892, that is
tailored for assessing posterior contraction rates in $L^r$-metrics, $1< r< +\infty$, taking an approach
that can only be used if one has sufficiently fine control of the approximation properties of the prior support 
in the $L^r$-metric considered.
\end{remark}

\begin{proposition}\label{prop:1}
Let $X_1,\,\ldots,\,X_n$ be i.i.d. observations from a density
$p_0\equiv p_{G_0}=G_0\ast f$, with the kernel density $f$ being the standard Laplace 
and the mixing distribution $G_0$ such that condition \eqref{eq:tailG0SS} holds as in Proposition \ref{prop:2}. If the baseline measure $\alpha$ of the Dirichlet process possesses continuous and positive 
density $\alpha'$ such that, for some constants $0<b<+\infty$ and $0<\tau\leq1$,
\begin{equation}\label{eq:tailG11}
\alpha'(y)\gtrsim \exp{(-b|y|^\tau)}\quad\mbox{for large $|y|$},
\end{equation}
then there exists a sufficiently large constant $M>0$ such that
\begin{equation}\label{eq:l2norm}
\Pi(||p_G-p_0||_2\geq M n^{-3/8}\log^{5/8}n
\mid \Data)=o_{\mathbf{P}}(1),
\end{equation}
where $\Pi(\cdot\mid \Data)$ denotes the posterior distribution corresponding 
to a Dirichlet-Laplace process mixture prior after $n$ observations.
\end{proposition}

As previously mentioned, convergence of the posterior distribution at
a rate $\varepsilon_n$ implies the existence of point estimators that converge at least as fast as $\varepsilon_n$ in the frequentist sense, see, for instance, Theorem 2.5 in Ghosal \emph{et al}.
(2000), p. 506, for the construction of a point estimator that applies to general statistical models and posterior distributions. The posterior expectation of the density $p_G$, which we refer to as the Bayes' density estimator,
$$\hat p_n^{\textrm{B}}(\cdot):=
\int_{\mathscr G} p_G(\cdot)\Pi(\d
 G\mid \Data),
$$
has a similar property when jointly considered with bounded semi-metrics that are convex or whose square is convex in one argument. When the random mixing distribution $\tilde G$ is distributed according to a Dirichlet process, 
the expression of the Bayes' density estimator $\hat p_n^{\textrm{B}}$
is given by formula (2.6) of Lo (1984), p. 353,
replacing $K(\cdot,\,u)$ with 
$\frac{1}{2}\exp{\{-|\cdot-u|\}}$ at each occurrence.

\begin{corollary}\label{cor:1}
Suppose that condition \eqref{eq:tailG0SS} holds for some decreasing function $A_0:\,(0,\,+\infty)\rightarrow [0,\,1]$ and a finite constant $c_0>(1/e)$ such that
\begin{equation}\label{eq:44}
G_0([-T,\,T]^c)\leq A_0(T)\lesssim \exp{(-e^{c_0T})}\quad\mbox{for large $T>0$}
\end{equation}
and condition \eqref{eq:tailG1} holds as in Proposition \ref{prop:2}. 
Then,
$$d(\hat p^{\mathrm{B}}_n,\,p_0)=O_{\mathbf{P}}(n^{-3/8}\log^{1/2} n),$$
for $d$ being either the Hellinger or the $L^1$-metric.
 \end{corollary}

\begin{proof}
In virtue of the inequality in \eqref{eq:L1_Hel}, it suffices to prove the assertion for the Hellinger metric. The proof follows standard arguments as, for instance, in Ghosal \emph{et al}. (2000), pp. 506--507. 
By convexity of $h^2$ in each argument and
Jensen's inequality, for $\varepsilon_n:=\max\{\bar \varepsilon_n,\, \tilde\varepsilon_n\}=n^{-3/8}(\log n)^{(3\vee 4)/8}=n^{-3/8}\log^{1/2} n
$ and a sufficiently large constant $M>0$,
\[\begin{split}
h^2(\hat p_n^{\textrm{B}},\,p_0)&\leq \int_{\mathscr G} h^2(p_G,\,p_0)\Pi(\d G\mid\Data)\\
&=\pt{\int_{h(p_G,\,p_0)<M\varepsilon_n}+\int_{h(p_G,\,p_0)\geq M\varepsilon_n}
} h^2(p_G,\,p_0)\Pi(\d G\mid\Data)\\[5pt]
&
\lesssim M^2\varepsilon_n^2 +  2 \Pi(h(p_G,\,p_0)
\geq M\varepsilon_n\mid \DataXe).
\end{split}\]
It follows that $$P_0^n h^2(\hat p_n^{\textrm{B}},\,p_0)
\lesssim M^2\varepsilon_n^2 +  2 P_0^n\Pi( h(p_G,\,p_0)
\geq M\varepsilon_n\mid \DataXe)\lesssim \varepsilon^2_n+o(\varepsilon_n^2)$$
because we can apply the almost sure version of Theorem 7 in Scricciolo (2007), p. 636 (see also Theorem A.1 in Scricciolo (2006), p. 2918), which, 
under the prior mass condition 
\begin{equation}\label{eq:74}
\Pi(h^2(p_G,\,p_0)\|p_0/p_G\|_\infty\leq \tilde\varepsilon_n
^2)\gtrsim \exp{(-Bn\tilde\varepsilon_n^2)},
\end{equation} 
with $\tilde\varepsilon_n:=n^{-3/8}\log^{1/2} n$ and a constant $0<B<+\infty$,
yields exponentially fast convergence of the posterior distribution since $P_0^n\Pi( h(p_G,\,p_0)
\geq M\varepsilon_n\mid \DataXe)\lesssim \exp{(-B_1n\tilde\varepsilon_n^2)}$ for a suitable 
constant $0<B_1<+\infty$. To verify that condition \eqref{eq:74} is satisfied, we can proceed as in the proof of Proposition \ref{prop:1}: for any $G$ satisfying \eqref{eq:condmixing}, not only is $h(p_G,\,p_0)\lesssim \varepsilon$, but, under assumption \eqref{eq:44} which guarantees that $M_{G_0}(-1)<+\infty$ and $M_{G_0}(1)<+\infty$, it also is
\[||p_0/p_G||_\infty\leq e^{a_\varepsilon}[M_{G_0}(-1)+M_{G_0}(1)]\lesssim\log(1/\varepsilon),\quad\mbox{ for $a_\varepsilon:= A_0^{-1}(\varepsilon^2)\lesssim\log\log(1/\varepsilon)$}.\]
Then, \[\log \Pi(h^2(p_G,\,p_0)\|p_0/p_G\|_\infty\leq \varepsilon^2\log(1/\varepsilon))\gtrsim -\varepsilon^{-2/3}\log(1/\varepsilon).\]
Condition \eqref{eq:74} is thus verified for $\tilde\varepsilon_n:=\varepsilon\log^{1/2}(1/\varepsilon)=n^{-3/8}\log ^{1/2}n$. Conclude that $h(\hat p^{\mathrm{B}}_n,\,p_0)=O_{\mathbf{P}}(\varepsilon_n)$.
\qed
\end{proof}

\begin{remark}
Admittedly, condition \eqref{eq:44} imposes a stringent constraint on the tail decay rate of $G_0$. An alternative sufficient condition for concluding that
\begin{equation}\label{eq:110}
P_0^n\Pi( d(p_G,\,p_0)
\geq M\varepsilon_n\mid \DataXe)=o(\varepsilon_n^2),\quad \mbox{ for \,$d=h$\, or \,$d=\|\cdot\|_1$,}
\end{equation}
is a prior mass condition involving the $k$th absolute moment of $\log(p_0/p_G)$ for a suitable value of $k$, in place of the sup-norm $\|p_0/p_G\|_\infty$, 
which can possibly induce a lighter condition on $G_0$. 
For $\tilde\varepsilon_n:=n^{-3/8}\log^{\omega}n$, with $\omega>0$, let
$\varepsilon_n:=\max\{\bar \varepsilon_n,\, \tilde\varepsilon_n\}=n^{-3/8}(\log n)^{(3/8)\vee \omega}$. It is known from Lemma 10 of Ghosal and van der Vaart (2007), p. 220, that 
if
\begin{equation}\label{eq:340}
\Pi(B_{\mathrm{KL}}(P_0;\,\tilde\varepsilon_n^k))\gtrsim\exp{(-Bn\tilde\varepsilon_n^2)},\quad k\geq2,
\end{equation}
then
\begin{equation}\label{eq:93}
P_0^n\Pi( d(p_G,\,p_0)
\geq M\varepsilon_n\mid \DataXe)\lesssim (n\tilde\varepsilon_n^2)^{-k/2}.
\end{equation}
Thus, if condition \eqref{eq:340} holds for some $k\geq6$ so that $(n\tilde\varepsilon_n^2)^{-k/2}=o(\varepsilon_n^2)$, the value $k=6$ would suffice for the purpose,  then condition \eqref{eq:110} is satisfied.
\end{remark}

We now state a result on the rate of convergence for the Bayes' estimator, denoted by $\hat G_n^{\textrm B}$, of the mixing distribution $G_0$ for the $L^1$-Wasserstein deconvolution of Dirichlet-Laplace mixtures.
The Bayes' estimator is the posterior expectation of the random 
probability measure $\tilde G$, that is, 
$\hat G_n^{\textrm B}(\cdot):=E[\tilde G(\cdot)\mid \Data]$
and its expression can be derived from the expression of the posterior distribution, cf. Ghosh and Ramamoorthi (2003), pp. 144--146. In order to state the result, 
let $M_{\hat G_n^{\textrm B}}(s):=\int_{-\infty}^{+\infty} e^{sy}\,\d\hat G_n^{\textrm B}(y)$, $s\in\mathbb{R}$, whose expression can be obtained from
formula (2.6) of Lo (1984), p. 353, replacing $K(x,\,u)$ with $e^{s u}$ at 
all occurrences ($s$ playing the role of $x$). 
  
\begin{proposition}\label{prop:4}
Suppose that the assumptions of 
Corollary \ref{cor:1} hold. If, in addition, $\bar\alpha:=\alpha/\alpha(\mathbb{R})$ 
has finite moment generating function on some interval $(-s_0,\,s_0)$, with $0<s_0<1$, and 
\begin{equation}\label{eq:ass1}
\forall\, 0<s<s_0,\quad
\limsup_{n\rightarrow +\infty}P_0^nM_{\hat G_n^{\mathrm{B}}}(-s)\leq M_{G_0}(-s)
\,\, \mbox{and} \,\,\limsup_{n\rightarrow
+\infty}P_0^nM_{\hat G_n^{\mathrm{B}}}(s)\leq M_{G_0}(s),
\end{equation}
then
\begin{equation}\label{eq:wass1}
W_1(\hat G_n^{\mathrm{B}},\,G_0)=O_{\mathbf{P}}(n^{-1/8}(\log n)^{2/3}).
\end{equation}
\end{proposition}

\begin{proof}
Let $\rho_n:=n^{-1/8}(\log n)^{2/3}$ and, for a suitable finite constant $c_1>0$, $M_n=c_1(\log n)$. Fix numbers $s$ and $u$ such that $0<u<s<s_0<1$. For sufficiently large constants $0<T, \,T',\,T''<+\infty$,
reasoning as in Lemma \ref{lem:2}, 
\[\begin{split}
P_0^n(W_1(\hat G_n^{\mathrm{B}},\,G_0)>T\rho_n)&\leq
P_0^n(h(\hat p^{\mathrm{B}}_n,\,p_0)>T'\rho_n^3(\log n)^{-3/2})\\
&\qquad \quad+
P_0^n(M_{\hat G_n^{\mathrm{B}}}(-s)+M_{\hat G_n^{\mathrm{B}}}(s)>T''e^{uM_n}\rho_n)=:P_1+P_2.
\end{split}\]
By Corollary \ref{cor:1}, $h(\hat p^{\mathrm{B}}_n,\,p_0)=O_{\mathbf{P}}(n^{-3/8}\log^{1/2}n)$. Hence, $P_1\rightarrow0$ as $n\rightarrow+\infty$.
By Markov's inequality, for some real $\nu>0$,
\[\begin{split}
P_2
&\lesssim  e^{-uM_n}\rho_n^{-1}
[P_0^nM_{\hat G_n^{\mathrm{B}}}(-s)+P_0^nM_{\hat G_n^{\mathrm{B}}}(s)]\\
&\lesssim \frac{1}{n^\nu}
[P_0^nM_{\hat G_n^{\mathrm{B}}}(-s)+P_0^nM_{\hat G_n^{\mathrm{B}}}(s)]
\rightarrow0 \quad\mbox{ as $n\rightarrow+\infty$}
\end{split}\]
by assumption (\ref{eq:ass1}). Thus, $P_2\rightarrow0$ as $n\rightarrow+\infty$. The assertion follows. 
\qed
\end{proof}

Some remarks are in order.
There are two main reasons why we focus on deconvolution in the $L^1$-Wasserstein metric. 
The first one is related to the inversion inequality in \eqref{eq:wasserstein}, 
where the upper bound on the $L^p$-Wasserstein metric, as a function of the order $p\geq1$, increases as $p$ gets larger, thus making it advisable 
to begin the analysis from the smallest value of $p$.
The second reason is related to the interpretation of the assertion in \eqref{eq:wass1}: 
the $L^1$-Wasserstein distance between any
two probability measures $G_1$ and $G_2$ on some Borel-measurable space $(\mathscr{Y},\,\mathscr{B}(\mathscr{Y}))$, $\mathscr Y\subseteq \mathbb{R}$, 
with finite first absolute moments, is by itself an interesting distance  
because it metrizes
weak convergence plus convergence of the first absolute moments, 
but it is even more interesting in view of the fact that,
letting $G_1^{-1}(\cdot)$ and $G_2^{-1}(\cdot)$ denote the left-continuous inverse or quantile functions, $G_i^{-1}(u):=\inf\{y\in\mathscr{Y}:\,G_i(y)\geq u\}$, $u\in(0,\,1)$, $i=1,\,2$,
it can be written as the $L^1$-distance between the quantile functions or, equivalently, as the $L^1$-distance between the cumulative distribution functions, 
\begin{equation}\label{eq:34}
W_1(G_1,\,G_2)
=\int_{0}^{1}|G_1^{-1}(u)-G_2^{-1}(u)|\,\d u=
\int_{\mathscr Y}|G_1(y)-G_2(y)|\,\d y=||G_1-G_2||_1,
\end{equation}
see, \emph{e.g.}, Shorack and Wellner (1986), pp. 64--66. The representation in \eqref{eq:34} was obtained by Dall'Aglio (1956). Thus, by rewriting $W_1(\hat G_n^{\mathrm{B}},\,G_0)$ 
as the $L^1$-distance between 
the c.d.f.'s $\hat G_n^{\mathrm{B}}$ and $G_0$, 
the assertion of Proposition \ref{prop:4},
$$W_1(\hat G_n^{\mathrm{B}},\,G_0)=||\hat G_n^{\mathrm{B}}-G_0||_1=O_{\mathbf{P}}(n^{-1/8}(\log n)^{2/3}),$$
becomes more transparent and meaningful.


\section{Rates of convergence for ML estimation and $L^1$-Wasserstein deconvolution of Laplace mixtures}\label{sec:MLE}
In this section, we first study the rate of convergence in the Hellinger metric for the MLE 
$\hat p_n$ of a Laplace mixture density $p_0\equiv p_{G_0}=G_0 \ast f$, with unknown mixing distribution $G_0\in\mathscr G$. We then derive the rate of convergence in the $L^1$-Wasserstein metric for the MLE $\hat G_n$ of the mixing distribution $G_0$, which corresponds to the MLE $\hat p_n$ of the mixed density $p_0$, by appealing to an inversion inequality relating the Hellinger distance between Laplace mixture densities to any $L^p$-Wasserstein distance, $p\geq1$, between the corresponding mixing distributions (see
Lemma \ref{lem:2} in Appendix D).

A MLE  
$\hat p_n$ of $p_0$ is a measurable function of the observations taking values in $\mathscr P:=\{p_G:\,G\in\mathscr G\}$ such that
\[\hat p_n\in \underset{p_G\in \mathscr P}{\arg\max}
\frac{1}{n}\sum_{i=1}^n\log p_G(X_i)=\underset{p_G\in \mathscr P}{\arg\max}
\int(\log p_G)\,\d {\mathbb P_n},\]
where ${\mathbb P_n}:={n}^{-1}\sum_{i=1}^n\delta_{X_i}$ is the empirical measure
associated with the random sample $X_1,\,\ldots,\,X_n$, namely,
the discrete uniform distribution on the sample values that puts mass $1/n$ on each one of the observations. We assume that the MLE exists, but do not require it to be unique,
see Lindsay (1995), Theorem 18, p. 112, for sufficient conditions ensuring uniqueness.

Results on rates of convergence in the Hellinger metric for the MLE of a density
can be found in Birg\'{e} and Massart (1993), Van de Geer (1993) and Wong and Shen (1995); 
it can, however, be difficult to calculate the
$L^2$-metric entropy \emph{with bracketing} of the square-root densities that is employed in these articles.
Taking instead into account that a mixture model 
$\{\int_{\mathscr Y}K(\cdot,\,y)\,\d G(y):\,G\in\mathscr G\}$ is the closure of the convex hull of the collection of kernels 
$\{K(\cdot,\,y):\,y\in\mathscr Y\subseteq\mathbb{R}\}$, which is typically a much smaller class, a bound on a form of metric entropy \emph{without bracketing} of the class of mixtures can be derived from a covering number of the class of kernels (a result on metric entropy \emph{without bracketing} of convex hulls that is deducible from Ball and Pajor (1990)), so that a relatively simple \vir{recipe} can be given
to obtain (an upper bound on) the rate of convergence in the Hellinger metric for the MLE of a density in terms of the \vir{dimension} of the class of kernels and the behaviour of $p_0$ near zero, cf. Corollary 2.3 of Van de Geer (1996), p. 298.

\begin{proposition}\label{prop:3}
Let the sampling density $p_0\equiv p_{G_0}=G_0\ast f$, with the kernel density $f$ 
being the standard Laplace and the mixing distribution $G_0\in\mathscr G$. 
Suppose that, for a sequence of non-negative 
real numbers $\sigma_n=O(n^{-3/8}\log^{1/8}n)$,
we have\smallskip
\begin{description}
\item[$(a)$] $\int_{p_0\leq \sigma_n}p_0\,\d \lambda\lesssim \sigma_n^2$,\\
\item[$(b)$] $\int_{p_0>\sigma_n}(1/p_0)\,\d \lambda\lesssim \log(1/\sigma_n)$.
\end{description}
Then, 
\begin{equation*}\label{eq:MLEhel}
h(\hat p_n,\, p_0)=O_{\mathbf{P}}(n^{-3/8}\log^{1/8}n).
\end{equation*}
\end{proposition}

\begin{proof}
We begin by spelling out the remark mentioned in the introduction concerning the fact that 
a mixture model is the closure of the convex hull of the collection of kernels. 
Recall that the convex hull of a class $\mathscr K$ of functions, 
denoted by $\mathrm{conv}(\mathscr K)$, is defined as the set of all finite convex combinations of functions in $\mathscr K$,
$$\mathrm{conv}(\mathscr K):=\Bigg\{\sum_{j=1}^r\theta_jK_j:\, \theta_j\geq0,\,K_j\in\mathscr K,\,j=1,\,\ldots,\,r,\,\sum_{j=1}^r\theta_j=1,\,r\in\mathbb{N}\Bigg\}.$$
In our case, $$\mathscr K:=\{f(\cdot-y):\,y\in\mathscr Y\subseteq\mathbb{R}\}$$ is the collection of kernels with $f$ the standard Laplace density.
The class $\mathscr P:=\{p_G:\,G\in\mathscr G\}$ of all Laplace convolution mixtures $p_G=G\ast f$ 
is the closure of the convex hull of $\mathscr K$,
$$\mathscr P=\overline{\mathrm{conv}}(\mathscr K).$$ Clearly, $\mathscr P$  is itself a convex class.
This remark enables us to apply Theorem 2.2 and Corollary 2.3 of Van de Geer (1996), pp. 297--298 and 310, or, equivalently, 
Theorem 7.7 of  Van de Geer (2000), pp. 104--105, whose conditions 
are hereafter shown to be satisfied. To the aim, we define the class
$${\mathscr K}/p_0:=\bigg\{\frac{f(\cdot-y)}{p_0(\cdot)}\1\{p_0>\sigma_n\}:\,y\in\mathscr{Y}\bigg\}$$
and the envelope function
$$\bar K(\cdot):=\sup_{y\in\mathscr{Y}}\frac{f(\cdot-y)}{p_0(\cdot)}\1\{p_0>\sigma_n\},$$
where we have suppressed the subscript $n$ in 
${\mathscr K}/p_0$ and $\bar K(\cdot)$
stressing possible dependence on $\sigma_n$ when $\sigma_n>0$.
Since, by assumption $(a)$, 
\[\int_{p_0\leq\sigma_n}\d P_0= \int_{p_0\leq\sigma_n}p_0\,\d \lambda\lesssim \sigma_n^2
\]
and, by assumption $(b)$, together with the fact that $\|f\|_\infty=1/2$,
\begin{equation}\label{eq:envelope}
\int\bar K^2\,\d P_0 \lesssim
\int_{p_0>\sigma_n}\frac{1}{p_0}\,\d \lambda\lesssim \log(1/\sigma_n),
\end{equation}
we can take the sequence $\delta_n^2\propto\sigma_n^2$ in condition (7.21) of Theorem 7.7 of Van de Geer (2000), p. 104.  
Because the (standard) Laplace kernel density $f$ is Lipschitz,
$$\forall\,y_1,\,y_2\in\mathscr Y,\quad|f(\cdot-y_1)-f(\cdot-y_2)|\leq\frac{1}{2} |y_1-y_2|,$$ see, \emph{e.g.}, Lemma A.1 in Scricciolo (2011),
pp. 299--300, on the set 
\begin{equation}\label{eq:set1}
\pg{\int \bar K^2\,\d {\mathbb P_n} \leq T^2\log(1/\delta_n)},
\end{equation}
where $T>0$ is a finite constant, we find that, 
for
$\d{\mathbb Q_n}:= \d {\mathbb P_n}/(T^2\log(1/\delta_n))$,
\[N(\delta,\,{\mathscr K}/p_0,\,||\cdot||_{2,{\mathbb Q_n}})\lesssim \delta^{-1}\quad\mbox{for }\delta>0,\]
where $||\cdot||_{2,{\mathbb Q_n}}$ denotes the $L^2({\mathbb Q_n})$-norm, that is, 
 $||g||_{2,{\mathbb Q_n}}:=(\int|g|^2\,\d {\mathbb Q_n})^{1/2}$.
So, in view of the result of Ball and Pajor (1990), reported as Theorem 1.1 in Van de Geer (1996), p. 295,
on the same set as in \eqref{eq:set1},
we have 
\[\log N(\delta,\,\overline{\mathrm{conv}}({\mathscr K}/p_0),\,||\cdot||_{2,{\mathbb Q_n}})\lesssim \delta^{-2/3},\]
hence
\[\log N(\delta,\,\overline{\mathrm{conv}}({\mathscr K}/p_0),\,||\cdot||_{2,{\mathbb P_n}})\lesssim \pt{\frac{T\log ^{1/2}(1/\delta_n)}{\delta}}^{2/3}.\]
Next, defined the class
\[\mathscr P^{(\mathrm{conv})}_{\sigma_n}:=\pg{\frac{2p_G}{p_G+p_0}\1\{p_0>\sigma_n\}:\,p_G\in\mathscr P}\]
considered in condition (7.20) of Theorem 7.7 in Van de Geer (2000), p. 104, 
since \[\log N(2\delta,\,\mathscr P^{(\mathrm{conv})}_{\sigma_n},\,||\cdot||_{2,{\mathbb P_n}})\leq
\log N(\delta,\,\overline{\mathrm{conv}}({\mathscr K}/p_0),\,||\cdot||_{2,{\mathbb P_n}}),\]
in view of \eqref{eq:envelope}, we have
\[\sup_{\delta>0}\frac{\log N(\delta,\,\mathscr P^{(\mathrm{conv})}_{\sigma_n},\,||\cdot||_{2,{\mathbb P_n}})}{H(\delta)}=O_{\mathbf{P}}(1)\]
for the non-increasing function of $\delta$
$$H(\delta):=\delta^{-2/3}\log^{1/3}(1/\delta_n),\quad\delta>0.$$ 
Taken $\Psi(\delta):=c_1\delta^{2/3}\log^{1/6}(1/\delta_n)$ with a suitable finite constant $c_1>0$, we have
$$\forall\,\delta\in(0,\,1),\,\,\,\Psi(\delta)\geq \pt{\int_{\delta^2/c}^\delta H^{1/2}(u)\,\d u} \vee \delta$$
and, for some $\varepsilon>0$, 
$\Psi(\delta)/\delta^{2-\varepsilon}$
is non-increasing. Then, for $\delta_n$ such that 
$\sqrt{n}\delta_n^2\geq \Psi(\delta_n)$, cf. condition (7.22) of Theorem 7.7 in Van de Geer (2000), p. 104, which implies that, consistently with the initial choice, we can take $\delta_n\propto n^{-3/8}\log ^{1/8}n$, we have 
$h(\hat p_n,\, p_0)=O_{\mathbf{P}}(\delta_n)$ and the proof is complete.
\qed
\end{proof}

\begin{remark}
If $p_0>0$ and $\mathscr Y$ is a compact interval $[-a,\,a]$, with $a>0$, then $h(\hat p_n,\,p_0)=O_{\mathbf{P}}(n^{-3/8})$. In fact, the sequence
$\sigma_n\equiv 0$, $||\bar K||_\infty\leq e^{2a}$ and $\int\bar K^2\,\d P_0 \leq e^{4a}$ so that, on the set 
$\{\int \bar K^2\,\d {\mathbb P_n}\leq T\}$, the entropy
$\log N(\delta,\,\overline{\mathrm{conv}}({\mathscr K}/p_0),\,||\cdot||_{2,{\mathbb P_n}})\lesssim \delta^{-2/3}$ and, reasoning as in Proposition \ref{prop:3}, we find the rate $n^{-3/8}$.
\end{remark}

We now derive a consequence of Proposition \ref{prop:3} on the rate of convergence in the $L^1$-Wasserstein metric for the MLE of $G_0$.
A MLE $\hat p_n$ of the \emph{mixed} density $p_0$ corresponds to a MLE $\hat G_n$ of the \emph{mixing} distribution $G_0$, that is, $\hat p_n\equiv p_{\hat G_n}$, such that
$$\hat G_n\in \underset{G\in \mathscr G}{\arg\max}
\frac{1}{n}\sum_{i=1}^n\log p_G(X_i)=\underset{G\in \mathscr G}{\arg\max}
\int(\log p_G)\,\d {\mathbb P_n}.$$
Clearly, $\hat G_n$ is a discrete distribution, but
we do not know the number of its components: Lindsay (1995) showed that the MLE $\hat G_n$ 
is a discrete distribution 
supported on at most $k\leq n$ support points, $k$ being the number of distinct observed values or data points.

\begin{corollary}\label{cor:wasserstein}
Suppose that the assumptions of 
Proposition \ref{prop:3} hold. If, in addition,
the mixing distribution $G_0$ has
finite moment generating function 
in some interval $(-s_0,\,s_0)$, with $0<s_0<1$, and 
\begin{equation}\label{eq:ass}
\forall\, 0<s<s_0,\quad
\limsup_{n\rightarrow +\infty}P_0^nM_{\hat G_n}(-s)\leq M_{G_0}(-s)
\quad \mbox{and} \quad \limsup_{n\rightarrow
+\infty}P_0^nM_{\hat G_n}(s)\leq M_{G_0}(s),
\end{equation}
where $M_{\hat G_n}(s):=\int_{\mathscr Y} e^{sy}\,\d \hat G_n(y)$, $s\in\mathbb{R}$, then
\[W_1(\hat G_n,\,G_0)=O_{\mathbf{P}}(n^{-1/8}(\log n)^{13/24}).\]
\end{corollary}

\begin{proof}
Let $k_n:=n^{-1/8}(\log n)^{13/24}$ and, for a suitable finite constant $c_2>0$, $M_n=c_2(\log n)$. Fix numbers $s$ and $u$ such that $0<u<s<s_0<1$. For sufficiently large constants $0<T,\,T',\,T''<+\infty$,
reasoning as in Lemma \ref{lem:2}, we have  
\[\begin{split}
P_0^n(W_1(\hat G_n,\,G_0)>T k_n)&\leq
P_0^n(h(\hat p_n,\, p_0)>T' k_n^3(\log n)^{-3/2})\\
&\qquad\quad+
P_0^n(M_{\hat G_n}(-s)+M_{\hat G_n}(s)>T''k_ne^{uM_n})=:P_1+P_2.
\end{split}\]
The term $P_1$ can be made arbitrarily small because $h(\hat p_n,\, p_0)=O_{\mathbf{P}}(n^{-3/8}\log^{1/8}n)$ by Proposition \ref{prop:3}. The term $P_2$ goes to zero as $n\rightarrow+\infty$: in fact, by Markov's inequality and assumption (\ref{eq:ass}), for some real $0<l<+\infty$,
\[\begin{split}
P_2&\lesssim e^{-uM_n}k_n^{-1}
[P_0^nM_{\hat G_n}(-s)+P_0^nM_{\hat G_n}(s)]\\
&\lesssim \frac{1}{n^l}
[P_0^nM_{\hat G_n}(-s)+P_0^nM_{\hat G_n}(s)]\rightarrow0 \quad \mbox{ as }n\rightarrow+\infty
\end{split}\]
and the assertion follows. 
\qed
\end{proof}

\begin{remark}
Assumption (\ref{eq:ass}) essentially requires that $M_{\hat G_n}$ is an asymptotically unbiased estimator of $M_{G_0}$ in some neighborhood of zero $(-s_0,\,s_0)$, with $0<s_0<1$.
An analysis of the asymptotic behaviour of certain linear functionals of the MLE $\hat G_n$ is 
presented in Van der Geer (1995), wherein sufficient conditions are provided so that they are $\sqrt{n}$-consistent, asymptotically normal and efficient.
\end{remark}


\section{Merging of Bayes and ML for $L^1$-Wasserstein deconvolution of Laplace mixtures}\label{sec:merging}
In this section, we show that the Bayes' estimator and the MLE of  
$G_0$ merge in the $L^1$-Wasserstein metric, their discrepancy 
vanishing, at worst, at rate $n^{-1/8}(\log n)^{2/3}$
because they both consistently estimate $G_0$ at a speed which is 
within a $(\log n)$-factor of 
$n^{-1/8}$, cf. Proposition \ref{prop:4} and Corollary \ref{cor:wasserstein}.

\begin{proposition}\label{prop:merging}
Under the assumptions of Proposition \ref{prop:4}  and Corollary \ref{cor:wasserstein}, 
we have
\begin{equation}\label{eq:merge}
W_1(\hat G_n^{\mathrm{B}},\,\hat G_n)=O_{\mathbf{P}}(n^{-1/8}(\log n)^{2/3}).
\end{equation}
\end{proposition}
\begin{proof}
By the triangle inequality, 
$$
W_1(\hat G_n^{\mathrm{B}},\,\hat G_n)\leq W_1(\hat G_n^{\mathrm{B}},\,G_0)+ W_1(G_0,\,\hat G_n),$$
where 
$W_1(\hat G_n^{\mathrm{B}},\,G_0)=O_{\mathbf{P}}(n^{-1/8}(\log n)^{2/3})$ and $W_1(G_0,\,\hat G_n)=O_{\mathbf{P}}(n^{-1/8}(\log n)^{13/24})$ by Proposition \ref{prop:4} and Corollary \ref{cor:wasserstein}, respectively. 
Relationship \eqref{eq:merge} follows. 
\qed
\end{proof}

Proposition \ref{prop:merging} states that the Bayes' estimator and the MLE of $G_0$ will eventually be indistinguishable and (an upper bound on) the speed of convergence for their $L^1$-Wasserstein discrepancy is determined by the stochastic orders of their errors in recovering $G_0$. 
The crucial question that remains open is whether the Bayes' estimator and the MLE are rate-optimal.
Concerning this issue, we note that, on the one hand, other deconvolution estimators for the distribution function attain the rate $n^{-1/8}$ when the error distribution is the standard Laplace, with the proviso, however, that the $L^1$-Wasserstein metric is not linked to the integrated quadratic risk between the c.d.f.'s used in the 
result we are going to mention, so that the rates are not comparable.
For instance, the estimator $G_n^{K}(h_n)(y):=\int_{-\infty}^yp_n^K(h_n)(u)\,\d u$, $y\in\mathbb{R}$, of the c.d.f. $G_0$
based on the standard deconvolution kernel density estimator  
is such that 
$\{\int_{-\infty}^{+\infty}E[G_n^{K}(h_n)(y)-G_0(y)]^2\,\d y\}^{1/2}=O(n^{-1/8})$ when 
no assumptions on $G_0$ are postulated, except for the existence of the first absolute moment, see (3.12) in Corollary 3.3 of Hall and Lahiri (2008), p. 2117.
On the other hand, a recent lower bound result, due to Dedecker \emph{et al}. (2015), Theorem 4.1, pp. 246--248, suggests that better rates are possible. 
For $M>0$ and $r\geq1$, let 
$\mathscr D(M,\,r)$ be the class of all probability measures $G$ on 
$(\mathbb{R},\,\mathscr{B}(\mathbb{R}))$ such that 
$\int_{-\infty}^{+\infty} |y|^r\,\d G(y)\leq M$. Let 
$f$ be the error density.  
Assume that there exist $\beta>0$ and $c>0$ such that, for every $\ell\in\{0,\,1,\,2\}$, it holds
$|\hat f^{(\ell)}(t)|\leq c(1+|t|)^{-\beta}$, $t\in\mathbb{R}$. 
Then, there exists a finite constant $C>0$ such that, for \emph{any} estimator ${\Hat{G}}_n$ (we warn the reader of the clash of notation with the symbol $\hat G_n$ previously used to denote the MLE of $G_0$),
$$\liminf_{n\to+\infty}n^{p/(2\beta+1)}\sup_{G\in \mathscr D(M,\,r)}EW_p^p({\Hat{G}}_n,\,G)>C.$$ For $p=1$ and the (standard) Laplace error distribution, this renders the lower bound $n^{-1/5}$, which is better than the leading term $n^{-1/8}$ of the upper bounds we have found, even if it is not said that either the Bayes' estimator or the MLE attains it.

Finally, a remark on the use of the term \vir{merging}. Even if this term 
is herein declined with a different meaning 
from that considered in Barron (1988), where merging is intended as the convergence to one of the ratio of the marginal likelihood to the joint 
density of the first $n$ observations, 
or from that in Diaconis and Freedman (1986), where
merging refers to the \vir{intersubjective agreement}, as more and more data become available, between two Bayesians with different prior opinions, the underlying idea is, in a broad sense, the same: different inferential procedures
become essentially indistinguishable for large sample sizes.


\section{Final remarks}\label{sec:finrmks}
In this note, we have studied rates of convergence for Bayes and maximum likelihood estimation of Laplace mixtures and for their $L^1$-Wasserstein deconvolution. 
The result on the convergence rate in the Hellinger metric for the MLE of Laplace mixtures is achieved taking a different approach from 
that adopted in Ghosal and van der Vaart (2001), which is
based on the $L^1$-metric entropy with bracketing  
of the set of densities under consideration and is difficult to apply in the present context, due to the non-analyticity of the Laplace density. 
Posterior contraction rates for Dirichlet-Laplace 
mixtures have been previously studied by Gao and van der Vaart (2016) in the case of compactly supported mixing distributions and have been here extended to mixing distributions with a possibly unbounded set of locations, this 
accounting for the derivation of more general entropy estimates, cf. Appendix B.
An interesting extension to pursue would be that of considering general kernel densities with polynomially decaying Fourier transforms in the sense of Definition \ref{def:algdecr}: indeed, in the proof of Proposition \ref{prop:1}, which gives an assessement of the posterior contraction rate in the $L^2$-metric for Dirichlet-Laplace mixtures, all conditions, except for the Kullback-Leibler prior mass requirement, hold for any kernel density as in Definition \ref{def:algdecr}, provided that $\beta>1$. The missing piece is an extension of Lemma 2 in Gao and van der Vaart (2016), pp. 615--616, which is preliminary for checking the Kullback-Leibler prior mass condition and guarantees that a Laplace mixture, with mixing distribution that is the re-normalized restriction of $G_0$ to a compact interval, can be approximated in the Hellinger metric by a Laplace mixture with a discrete mixing distribution having a sufficiently restricted number of support points. 
We believe that, as for the Laplace kernel, the number of support points of the approximating mixing distribution will ultimately depend only on the decay rate of the Fourier transform of the kernel density, even though, in a general proof, the explicit expression of the kernel density cannot be exploited as in the Laplace case. Extending the result on posterior contraction rates to general kernel mixtures would be of interest in itself and for extending the $L^1$-Wasserstein deconvolution result, even though this would pose in more general terms the rate-optimality question, as it happens for the $n^{-1/8}$-rate in the Laplace case, see the remarks at the end of Sect. \ref{sec:merging}.
We hope to report on these issues in a follow-up contribution.

\bigskip

\noindent{\small\bf{Acknowledgements}}\hspace*{0.3cm}The author
would like to thank the Editor and an anonymous Referee for their careful reading of the manuscript and helpful comments that have led to an improved presentation of the results.  
She gratefully acknowledges financial support from MIUR, grant n$^\circ$ 2015SNS29B \vir{Modern Bayesian nonparametric methods}.

\section*{Appendix A: Auxiliary results}
\begin{theopargself}
In this section, a sufficient condition on a convolution kernel $K\in L^1(\mathbb{R})$ is stated
in terms of its Fourier transform $\hat K$ so that the exact order of the $L^2$-norm error for approximating any probability density $f$, with polynomially decaying characteristic function $\hat f$ of degree $\beta>1/2$ 
(see  Definition \ref{def:algdecr} below)
by its convolution with $K_h:=h^{-1}K(\cdot/h)$, that is, by $f\ast K_h$, is assessed in terms of the bandwidth $h$. The
result is instrumental to the proof of Proposition \ref{prop:1} to show that any mixture density $p_G=G\ast f$, irrespective of the mixing distribution $G\in\mathscr G$, verifies the \emph{bias} condition $||p_G\ast K_h-p_G||_2=O(h^{\beta-1/2})$, which is involved in the definition of the sieve set in (15) of Theorem 2 in Gin\'{e} and Nickl (2011), p. 2891. We refer to the difference $(p_G\ast K_h-p_G)$ as the \emph{bias} because it is indeed the bias of the kernel density estimator $p_n^K(h):=\mathbb{P}_n\ast K_h$, when the observations are sampled from $p_G$: in fact, the bias $b[p_n^K(h)]:=E[p_n^K(h)]-p_G=p_G\ast K_h-p_G$. The condition in \eqref{eq:integrability} below, which traces back to Watson and Leadbetter (1963), see the first Theorem of Sect. 3B, pp. 486--487, is verified for any kernel $K$ of order $r$ greater than or equal to $\beta$, as later on spelled out in Remark \ref{rem:1}.

\begin{definition}\label{def:algdecr}
Let $f$ be a probability density function on
$\mathbb{R}$. The Fourier transform of $f$ or the characteristic function of the corresponding probability measure on $(\mathbb{R},\,\mathscr{B}(\mathbb{R}))$, denoted by $\hat f$,
is said to decrease algebraically of degree $\beta>0$ if there exists a constant $0<B_f<+\infty$ such that
\begin{equation}\label{eq:algebraic}
\lim_{|t|\rightarrow+\infty}|t|^\beta|\hat f(t)|=B_f.
\end{equation}
\end{definition}
Relationship \eqref{eq:algebraic} describes the tail behaviour of $|\hat f|$ by stating that it decays polynomially as $|t|^{-\beta}$. The class of probability measures on $(\mathbb{R},\,\mathscr{B}(\mathbb{R}))$ that have characteristic functions satisfying condition (\ref{eq:algebraic}) 
includes\\[-0.5cm]
\begin{itemize}
\item any gamma distribution with shape and scale parameters $\nu>0$ and $\lambda>0$, respectively,
whose characteristic function has expression $(1+ it/\lambda)^{-\nu}$, the role of $\beta$ in \eqref{eq:algebraic}
being played by $\nu$;\\[-0.33cm]
\item any distribution with characteristic function
$(1+|t|^\alpha)^{-1}$, $t\in\mathbb{R}$, for $0<\alpha \leq 2$,
which is called an $\alpha$-\emph{Laplace distribution} or \emph{Linnik's distribution}, cf. Devroye (1990);
the case $\alpha=2$ renders the characteristic function of a standard Laplace distribution. 
The role of $\beta$ in \eqref{eq:algebraic} is played by $\alpha$;\\[-0.33cm]
\item any distribution with characteristic function $(1+|t|^\alpha)^{-1/\beta}$, which, for $\beta=1$,
reduces to that of an $\alpha$-Laplace distribution. The exponent $\alpha/\beta$ plays the role of the polynomial's degree $\beta$ in \eqref{eq:algebraic}. Devroye (1990) observes that,
if $S_\alpha$ is any symmetric stable r.v. with characteristic function
$e^{-|t|^\alpha}$, $0<\alpha\leq2$, and $V_\beta$ is an independent r.v. with density
$e^{-v^\beta}/\Gamma(1+1/\beta)$, $v>0$, then the r.v. $S_\alpha V_\beta^{\beta/\alpha}$
has characteristic function $(1+|t|^\alpha)^{-1/\beta}$.  
\end{itemize}

\begin{lemma}\label{lem:1}
Let $f\in L^2(\mathbb{R})$ be a probability density function with Fourier transform $\hat f$ satisfying condition \eqref{eq:algebraic} for some $\beta>1/2$ and a constant $0<B_f<+\infty$. If $K
\in L^1(\mathbb{R})$
has Fourier transform $\hat K$ such that $\hat K(0)=1$ and
\begin{equation}\label{eq:integrability}
I^2_\beta[\hat K]:=\int_{\{t\neq0\}}\frac{|1-\hat K(t)|^2}{|t|^{2\beta}}\,\d t<+\infty,
\end{equation}
then $$h^{-2(\beta-1/2)}\|f-f\ast K_h \|_2^2\rightarrow \frac{1}{2\pi}\times
B^2_f\times
I^2_\beta[\hat K                     
] \quad\mbox{as } h
\rightarrow0.$$
\end{lemma}


\begin{proof}
\smartqed
Since it is assumed that $f\in L^1(\mathbb{R})\cap L^2(\mathbb{R})$, then $\hat f\in L^2(\mathbb{R})$ and necessarily $\beta>1/2$. 
Also, as
$K\in L^1(\mathbb{R})$, then $\|f\ast K_h\|_p\leq \|f\|_p\|K_h\|_1<+\infty$
for $p=1,\,2$.
Thus, $(f-f\ast K_h)\in L^1(\mathbb{R})\cap L^2(\mathbb{R})$ and, by Plancherel's Theorem, $\|f-f\ast K_h\|_2^2=(2\pi)^{-1}\|\hat f-\hat f \times
\hat K_h\|_2^2$. By the change of variable $z=ht$,
\[\begin{split}
\|f-f\ast K_h\|_2^2&=
\frac{1}{2\pi}\int_{-\infty}^{+\infty}
|\hat f(t)|^2|1-\hat K(h t)|^2\,\d t\\
&=\frac{1}{2\pi}h^{2(\beta-1/2)}
\Bigg\{B_f^2 \times I^2_\beta[\hat K]
+
\int_{\{z\neq0\}}\frac{|1-\hat K(z)|^2}{|z|^{2\beta}}\Big
[|z/h|^{2\beta}|\hat f(z/h)|^2-B_f^2\Big]
\,\d z\Bigg\},
\end{split}
\]
where, for every sequence of positive real numbers $h_n\rightarrow 0$, the integral on the right-hand side of the last display tends to zero by the dominated convergence theorem
due to assumption (\ref{eq:integrability}). The assertion follows. 
\qed
\end{proof}

In the following remark, which is essentially due to Davis~(1977), cf. Sect. 3, pp.~532--533, sufficient conditions on a kernel $K\in L^1(\mathbb{R})$ are given so that $\hat K(0)=1$ and the requirement in \eqref{eq:integrability} is satisfied.
The conditions in \eqref{eq:3} below require that $K$ is a \emph{kernel of order $r\geq\beta>1/2$}, the order of a kernel being the first non-zero \vir{moment} of the kernel,
cf. Definition 1.3 in Tsybakov (2004), p. 5. 

\begin{remark}\label{rem:1}
For $K\in L^1(\mathbb{R})$, the Fourier transform $\hat K$ is continuous and bounded so that the integral $\int_{-\infty}^{+\infty}|t|^{-2\beta}|1-\hat K(t)|^2\1_{[1,\,+\infty)}(|t|)\,\d t<+\infty$ for $\beta>1/2$. The problem with condition \eqref{eq:integrability} is therefore the integrability of the function $t\mapsto |t|^{-2\beta}|1-\hat K(t)|^2$ for $|t|\in(0,\,1
)$.
Suppose that
\begin{eqnarray}\label{eq:3}
&&
\hspace*{-0.5cm}\int_{-\infty}^{+\infty} K(x)\,\d x=1, \nonumber\\
&&\hspace*{-0.5cm}
\mbox{$\exists\,r\in\mathbb{N}$, $r\geq\beta>\frac{1}{2}$\,:}\int_{-\infty}^{+\infty} x^j K(x)\,\d x=0\,\,\mbox{ for $j=1,\,\ldots,\,r-1$\,\, only if\,\, $r\geq2$,}\nonumber\\ 
&&\hspace*{-0.7cm}\mbox{and }\hspace*{5cm} \int_{-\infty}^{+\infty} x^r K(x)\,\d x\neq 0
\end{eqnarray}
and
\begin{equation}\label{eq:45}
\int_{-\infty}^{+\infty} |x|^r |K(x)|\,\d x<+\infty,
\end{equation}
(the value $r$ being called the \emph{characteristic exponent} of $\hat K$, see Parzen (1962), pp. 1072--1073), then 
\[\hat K(0)=1 \,\,\mbox{ and }\,\, \int_{-\infty}^{+\infty}|t|^{-2\beta}|1-\hat K(t)|^2\1_{(0,\,1)}(|t|)\,\d t<+\infty.\] 
In fact, $\hat K(0)=\int_{-\infty}^{+\infty}K(x)\,\d x=1$. Also, for every real number $t\neq0$,
\[\begin{split}
\frac{1-\hat K(t)}{t^{r}} = - \frac{\hat K(t)-1}{t^{r}}
&=-\frac{1}{t^{r}}\int_{-\infty}^{+\infty} (e^{itx}-1)K(x)\, \d x\\
&=- \frac{1}{t^r} \int_{-\infty}^{+\infty} \Bigg[e^{itx}-\sum_{j=0}^{r-1}\frac{(itx)^j}{j!}\Bigg]K(x)\, \d x\\
&= - \frac{i^r}{(r-1)!}\int_{-\infty}^{+\infty}  x^r K(x)\int_0^1(1-u)^{r-1}e^{itux}\,\d u\,\d x.
\end{split}
\]
By the dominated convergence theorem, condition \eqref{eq:45} implies that
\[\frac{1-\hat K(t)}{t^r}\rightarrow
-\frac{i^r}{r!}\int_{-\infty}^{+\infty} x^r K(x)\,\d x \quad\mbox{as $t\rightarrow0$,}\]
where the limit is non-zero in virtue of the last condition on the right-hand side of \eqref{eq:3}. 
It is seen by comparison that, since $r\geq\beta$, the integral $\int_{-\infty}^{+\infty}|t|^{-2\beta}|1-\hat K(t)|^2\1_{(0,\,1)}(|t|)\,\d t<+\infty$  and condition \eqref{eq:integrability} is satisfied. If, for instance, $1/2<\beta\leq 2$, then any symmetric probability density $K$ on $\mathbb{R}$, with finite, non-zero second moment $\mu_2:=
\int_{-\infty}^{+\infty} x^2 K(x)\,\d x\neq 0$ is such that
$I^2_\beta[\hat K]<+\infty$.
\end{remark}
\end{theopargself}


\section*{Appendix B: Entropy estimates}
\begin{theopargself}
In this section, Hellinger and $L^1$-metric entropy estimates
for a class of Laplace mixture densities, with mixing distributions having tails dominated by a given decreasing function, are provided.
The result of Lemma \ref{lem:entropy} extends, along the lines of Theorem 7 in Ghosal and van der Vaart (2007), pp. 708--709, Proposition 2 of Gao and van der Vaart (2016), p. 617, which deals with Laplace mixtures having compactly supported mixing distributions. Lemma \ref{lem:entropy} is invoked in the proof of Proposition \ref{prop:2}, reported in Appendix C, to verify that the entropy condition is satisfied.

\begin{lemma}\label{lem:entropy} 
For a given decreasing function $A:\,(0,\,+\infty)\rightarrow[0,\,1]$, with inverse $A^{-
1}$, define the class of Laplace mixture densities
\begin{equation*}\label{eq:set}
\mathscr P_A:=\{p_G:\,G([-a,\,a]^c)\leq A(a) \,\mbox{ for all } a>0\}.
\end{equation*}
Then, for every $0<\varepsilon<1$,
\begin{itemize}
\item taking $a\equiv a_\varepsilon:=A^{-1}(\varepsilon)$ in the definition of $\mathscr P_A$, we have
\begin{equation}\label{eq:entropyL1}
\mbox{ }
\log N(3\varepsilon,\,\mathscr P_A,\,||\cdot||_1)\lesssim \varepsilon^{-2/3}\log \frac{A^{-1}(\varepsilon)}{\varepsilon^2},
\end{equation}
\item taking $a\equiv a_{\varepsilon^2}:=A^{-1}(\varepsilon^2)$ in the definition of $\mathscr P_A$, we have
\begin{equation}\label{eq:entropyHel}
\log N((\sqrt{2}+1)\varepsilon,\,\mathscr P_A,\,h)
\lesssim \varepsilon^{-2/3}\log \frac{A^{-1}(\varepsilon^2)}{\varepsilon^2}.
\end{equation}
\end{itemize}
\end{lemma}


\begin{proof}
Concerning the $L^1$-metric entropy in \eqref{eq:entropyL1}, since $a\equiv a_\varepsilon:=A^{-1}(\varepsilon)$ satisfies $G([-a_\varepsilon,\,a_\varepsilon]^c)\leq A(a_\varepsilon)=\varepsilon$
for all $G$ as in the definition of $\mathscr P_A$, Lemma A.3 of Ghosal and van der Vaart (2001), p. 1261, implies that
the $L^1$-distance between any density $p_G\in \mathscr P_A$ and the corresponding density $p_{G^\ast}$, with mixing distribution $G^\ast$ defined as the re-normalized restriction of $G$ to
$[-a_\varepsilon,\,a_\varepsilon]$, is bounded above by $2\varepsilon$.
Then, in virtue of the inequality in \eqref{eq:L1_Hel},
a Hellinger $(\varepsilon/2)$-net over the class of densities $\mathscr P_{a_\varepsilon}:=\{p_G:\,G([-a_\varepsilon,\,a_\varepsilon])=1\}$ is an $L^1$-metric $3\varepsilon$-net over $\mathscr P_A$, where
$$\log N\big(\varepsilon/2,\,\mathscr P_{a_\varepsilon},\,h\big)\lesssim \varepsilon^{-2/3}\log \frac{a_\varepsilon}{\varepsilon^2}$$ by Proposition 2 of Gao and van der Vaart (2016), p. 617.
The inequality in \eqref{eq:entropyL1} follows.

Concerning the Hellinger-metric entropy in \eqref{eq:entropyHel}, by taking
$a\equiv a_{\varepsilon^2}:=A^{-1}(\varepsilon^2)$,
for every $p_G\in \mathscr P_A$ and the corresponding $p_{G^\ast}$, with mixing distribution $G^\ast$ defined as the re-normalized restriction of $G$ to
$[-a_{\varepsilon^2},\,a_{\varepsilon^2}]$, by the inequality in \eqref{eq:Hel_L1}, we have
$h^2(p_G,\,p_{G^\ast})\leq ||p_G-p_{G^\ast}||_1\leq 2G([-a_{\varepsilon^2},\,a_{\varepsilon^2}]^c)\leq 2\varepsilon^2$,
which implies that $h(p_G,\,p_{G^\ast})\leq \sqrt{2}\varepsilon$. Thus, a Hellinger $\varepsilon$-net
over $\mathscr P_{a_{\varepsilon^2}}:=\{p_G:\,G([-a_{\varepsilon^2},\,a_{\varepsilon^2}])=1\}$ is a $(\sqrt{2}+1)\varepsilon$-net over $\mathscr P_A$, where
$$\log N\big(\varepsilon,\,\mathscr P_{a_{\varepsilon^2}},\,h\big)\lesssim \varepsilon^{-2/3}\log \frac{a_{\varepsilon^2}}{\varepsilon^2}$$
again by Proposition 2 of Gao and van der Vaart (2016), p. 617. The inequality in \eqref{eq:entropyHel} follows.
\qed
\end{proof}
\end{theopargself}


\section*{Appendix C: Posterior contraction rates in $L^r$-metrics, $1\leq r\leq 2$, for Dirichlet-Laplace mixtures}
\label{appendix:rates}
\begin{theopargself}
In this section, we prove Proposition \ref{prop:2} and Proposition \ref{prop:1} of Sect. \ref{sec:Bayes} on contraction rates in the $L^1$ and $L^2$-metrics, respectively, for
the posterior distribution corresponding to a Dirichlet process mixture of Laplace densities.

\medskip

\noindent\emph{Proof of Proposition \ref{prop:2}}
In order to derive the Hellinger or the $L^1$-metric posterior contraction rate,
we can appeal to Theorem 2.1 of Ghosal \emph{et al}. 
(2000), p. 503, or Theorem 2.1 of Ghosal and van der Vaart (2001), p. 1239.
We define a sieve set for which conditions (2.2) or (2.8) and (2.3) or (2.9), postulated in the aforementioned theorems, are satisfied.
To the aim, once recalled that $\alpha(\mathbb{R})<+\infty$, let $\bar\alpha:=\alpha/\alpha(\mathbb{R})$ be the
probability measure corresponding to the baseline measure $\alpha$ 
of the Dirichlet process.
Consistently with the notation adopted throughout, $\bar\alpha$ is also used 
to denote the
corresponding cumulative distribution function.
By a result of Doss and Sellke (1982), p. 1304, which concerns the tails of probability measures chosen from a Dirichlet prior, we have that, for almost every sample distribution $G$, if $a>0$ is large enough so that $\bar\alpha(-a)=1-\bar\alpha(a)$ is sufficiently small,
then
\[
\begin{split}
G([-a,\,a]^c)&\leq G(-a)+1-G(a)\\
&\leq \exp{\bigg\{-\frac{1}{\bar\alpha(-a)|\log \bar\alpha(-a)|^2}\bigg\}}+
\exp{\bigg\{-\frac{1}{[1-\bar\alpha(a)]|\log[1-\bar\alpha(a)]|^2}\bigg\}}\\
&=2\exp{\bigg\{-\frac{1}{\bar\alpha(-a)\,|\log \bar\alpha(-a)|^2}\bigg\}}\\
&< A_\eta(a),
\end{split}
\]
having set the position $A_\eta(a):=2\exp{\{-[\bar\alpha(-a)]^{-\eta}\}}$
for some fixed $0<\eta<1$. The inverse function $A_\eta^{-1}:\,(0,\,1)\rightarrow (0,\,+\infty)$
is defined
as $A^{-1}_\eta:\,u\mapsto -\bar\alpha^{-1}(\log^{-1/\eta}(2/u))$,
where the function $\bar\alpha^{-1}(\cdot)$ is the left-continuous inverse of $\bar\alpha(\cdot)$, that is, $\bar\alpha^{-1}(u):=\inf\{y\in\mathbb{R}:\,\bar\alpha(y)\geq u\}$, $u\in(0,\,1)$.
Considered the class of densities
$\mathscr P_{A_\eta}:=\{p_G:\,G([-a,\,a]^c)\leq A_\eta(a) \,\mbox{ for all } a>0\}$,
we have
$\Pi(\mathscr P_{A_\eta})=1$.
For any sequence of positive real numbers $\bar\varepsilon_n\downarrow0$,
set the position $a\equiv a_{\bar\varepsilon_n}:=A_\eta^{-1}(\bar\varepsilon_n)$ 
and defined the sieve set
$\mathscr P_n:=\{p_G:\,G([-a_{\bar\varepsilon_n},\,a_{\bar\varepsilon_n}]^c)\leq A_\eta(a_{\bar\varepsilon_n})=\bar\varepsilon_n\}$, we have $$\Pi(\mathscr P\setminus \mathscr P_n)=0$$ and condition (2.3) or (2.9) is satisfied.
As for condition (2.2) or (2.8), taking $\bar\varepsilon_n=n^{-3/8}\log^{3/8}n$, by Lemma \ref{lem:entropy}, we have
\begin{equation}\label{eq:entropy bound}
\log D(\bar\varepsilon_n,\,\mathscr P_n,\, ||\cdot||_1)\leq
\log N(\bar\varepsilon_n/2,\,\mathscr P_n,\, ||\cdot||_1)\lesssim
(\bar\varepsilon_n)^{-2/3}\log \frac{A_\eta^{-1}(\bar\varepsilon_n/6)}{\bar\varepsilon_n^2}\lesssim n\bar\varepsilon_n^2.
\end{equation}
The same bound as in \eqref{eq:entropy bound} also holds for the Hellinger metric entropy.
The Kullback-Leibler prior mass condition (2.4) of Theorem 2.1 of Ghosal \emph{et al}.
(2000), p. 503, or, equivalently, condition (2.10) of Theorem 2.1 of Ghosal and van der Vaart (2001),
p. 1239, can be seen to be satisfied for $\tilde\varepsilon_n:=n^{-3/8}\log^{5/8}n$. For the verification of this condition, we refer the reader to 
condition (2) of Proposition \ref{prop:1} below,
whose requirement \eqref{eq:tailG11} is satisfied
under assumption \eqref{eq:tailG1} of Proposition \ref{prop:2}. 
The proof is completed by taking
$\varepsilon_n:=\max\{\bar\varepsilon_n,\,\tilde\varepsilon_n\}=n^{-3/8}\log^{5/8}n$. For the sake of clarity, we remark that the role of $\tilde\varepsilon_n$ is played by $\varepsilon_n$ in the proof of Proposition 
\ref{prop:1}.
\qed

\medskip

We now prove Proposition \ref{prop:1} on the posterior contraction rate in the $L^2$-metric.
The result relies on Theorem 3 of Gin\'{e} and Nickl (2011), p. 2892, which gives sufficient
conditions for deriving posterior contraction rates in $L^r$-metrics, $1<r<+\infty$. All assumptions of Theorem 3, except for condition (2), are shown to be satisfied for any kernel density $f$ as in Definition \ref{def:algdecr} with $\beta>1$.
This includes the (standard) Laplace kernel density as a special case when $\beta=2$.
Condition (2), which requires the prior mass in Kullback-Leibler type neighborhoods of the sampling density
$p_0\equiv p_{G_0}=G_0\ast f$ to be not exponentially small, relies on a preliminary approximation result of the density $p_{G_0^*}=G_0^*\ast f$, with mixing distribution $G_0^*$ obtained as
the re-normalized restriction of $G_0$ to a compact interval, 
by a mixture density that has a discrete mixing distribution with a sufficiently restricted number of support points.
This result is known to hold for the Laplace kernel density in virtue of
Lemma 2 of Gao and van der Vaart (2016), pp. 615--616.

\medskip

\noindent\emph{Proof of Proposition \ref{prop:1}}
We apply Theorem 3 of Gin\'{e} and Nickl (2011), p. 2892, with $r=2$.
We refer to the conditions of this theorem using the same letters/numbers 
as in the original article.
Let $\gamma_n\equiv 1$ and
$\delta_n\equiv\varepsilon_n:=n^{-3/8}\log^{5/8}n$, $n\in\mathbb{N}$. 
\begin{itemize}
\item \emph{Verification of condition} (b)\\[2pt]
Condition (b), which requires that $\varepsilon_n^2=O(n^{-1/2})$, is satisfied in 
the general case for $\varepsilon_n=n^{-(\beta-1/2)/2\beta}\log ^\kappa n$, with some $\kappa >0$
and $\beta>1$.\\
\item \emph{Verification of condition} (1)\\[2pt]
Condition (1) requires that the prior probability of the complement of a sieve set $\mathscr P_n$ is exponentially small.
We show that, in the present setting, the prior probability of a sieve set $\mathscr P_n$, chosen
as prescribed by (15) in Theorem 2 of Gin\'{e} and Nickl (2011), p. 2891, is equal to zero.
Let $J_n$ be any sequence of positive real numbers satisfying
$2^{J_n}\leq c n\varepsilon_n^2
$ for some fixed constant $0<c<+\infty$.
Let $K$ be a convolution kernel such that it is of bounded $p$-variation for some finite real number $p\geq1$, right (or left) continuous and satisfies $||K||_\infty<+\infty$,
$\int_{-\infty}^{+\infty}(1+|z|)^w|K(z)|\,\d z<+\infty$ for some $w>2$, $\hat K(0)=1$ and
$I^2_\beta[\hat K
]<+\infty$, cf. condition \eqref{eq:integrability} in Lemma \ref{lem:1}.
Defined the sieve set
$$\mathscr P_n:=\big\{p_G\in\mathscr P:\,||p_G\ast K_{2^{-J_n}}-p_G||_2\leq C\delta
_n\big
\},$$
where $K_{2^{-J_n}}(\cdot):=2^{J_n}K(\cdot2^{J_n})$ and $C>0$ is a finite constant depending only
on $K$ and $f$, we have
\begin{equation*}\label{eq:sieveprob}
\Pi(\mathscr P\setminus \mathscr P_n)=0\quad\mbox{for all $n\in\mathbb{N}$.}
\end{equation*}
In fact, for every $G\in \mathscr G$,
by Plancherel's Theorem,
$||p_G\ast K_{2^{-J_n}}-p_G||_2^2=||p_G-p_G\ast K_{2^{-J_n}}||_2^2=(2\pi)^{-1}
||\hat p_G- \hat p_G\times \hat K
_{2^{-J_n}}||_2^2\leq (2\pi)^{-1}||\hat f-
\hat f \times\hat K
_{2^{-J_n}}||_2^2$
and, by Lemma \ref{lem:1}, $||\hat f-
\hat f \times\hat K
_{2^{-J_n}}||_2^2
\sim (2^{-J_n})^{2\beta-1} \times B_f^2\times I_\beta^2[\hat K]$,
where, for $\beta=2$, we have $(2^{-J_n})^{2\beta-1}=(2^{-J_n})^{3}=
O(\delta_n^2)$.
Thus,
\begin{equation}\label{eq:sieve}
\forall\,G\in\mathscr G,\,\,\,
||p_G\ast K_{2^{-J_n}}-p_G||_2=O(\delta_n)
\end{equation}
and condition (1) is verified. Relationship
\eqref{eq:sieve} holds, in particular, for $p_0\equiv p_{G_0}=G_0\ast f$.
Furthermore, $p_0\in L^2(\mathbb{R})$ if $f\in L^2(\mathbb{R})$, which is the case for the (standard) Laplace kernel density, because
$||p_0||_2^2=(2\pi)^{-1}||\hat p_0||_2^2
\leq (2\pi)^{-1}||\hat f ||_2^2= ||f ||_2^2<+\infty$. \\

\item \emph{Verification of condition} (2)\\[2pt]
Condition (2) requires that, for some finite constant $C_1>0$, the prior probability of Kullback-Leibler type
neighborhoods of $P_0$ of radius $\varepsilon_n^2$ is at least $\exp{(-C_1 n\varepsilon_n^2)}$, that is,
$\Pi(B_{\textrm{KL}}(P_0;\,\varepsilon_n^2)
)\gtrsim \exp{(-C_1 n\varepsilon_n^2)}$.
Fix $0<\varepsilon\leq (1-e^{-1})/\sqrt{2}$ and
let $a_\varepsilon:=A_0^{-1}(\varepsilon^2)$, where $A_0^{-1}$ is the inverse of the function $A_0$
in condition \eqref{eq:tailG0SS}.
Define $G_0^\ast$ as the re-normalized restriction of $G_0$ to $[-a_\varepsilon,\,a_
\varepsilon]$. By Lemma A.3 of Ghosal and van der Vaart (2001), p. 1261,
and assumption \eqref{eq:tailG0SS}, we have
$||p_{G_0}-p_{G_0^\ast}||_1\leq2 G_0([-a_\varepsilon,\,a_
\varepsilon]^c)\lesssim \varepsilon^2$. From the 
inequality in \eqref{eq:Hel_L1},
$h^2(p_{G_0},\,p_{G_0^\ast})\leq ||p_{G_0}-p_{G_0^\ast}||_1\lesssim \varepsilon^2$, 
whence $h(p_{G_0},\,p_{G_0^\ast})\lesssim\varepsilon$.
It is known from Lemma 2 of Gao and van der Vaart (2016), pp. 615--616,
that there exists a discrete distribution $G_0'$ such that $h(p_{G_0'},\,p_{G_0^\ast})\lesssim \varepsilon$. The distribution $G_0'$ has at most $N\asymp \varepsilon^{-2/3}$ support points $y_1,\,\ldots,\,y_N$ in $[-a_\varepsilon,\,a_\varepsilon]$, which we may assume to be at least $2\varepsilon^2$-separated. If not, we can take a maximal $2\varepsilon^2$-separated set in the support points of $G_0'$ and replace $G_0'$  with the discrete
distribution $G_0''$ obtained by relocating the masses of $G_0'$ to the nearest points of the $2\varepsilon^2$-net. Then, 
$h^2(p_{G_0'},\,p_{G_0''}
)\lesssim \max_{1\leq j
\leq N}|y_j'-y_j''|\lesssim \varepsilon^2$,
as shown in Proposition 2 of Gao and van der Vaart (2016), p. 617.
Let $G_0'=\sum_{j=1}^Np_j\delta_{y_j}$, with $|y_j-y_k|\geq2\varepsilon^2$ for all $1\leq j\neq k\leq N$. For any distribution $G$ 
such that
\begin{equation}\label{eq:condmixing}
\sum_{j=1}^N|G([y_j-\varepsilon^2,\,y
_j+\varepsilon^2])-p_j|\leq \varepsilon^2,
\end{equation}
we have
$||p_G-p_{G_0'}||_1\lesssim \varepsilon^2$ 
by Lemma 5 of Gao and van der Vaart (2016), p. 620.
Thus,
\[\begin{split}
h^2(p_G,\,p_{G_0}) &\lesssim
h^2(p_G,\,p_{G_0'}) + h^2(p_{G_0'},\,p_{G_0^\ast}) + h^2(p_{G_0^\ast},\,p_{G_0})\\
&\lesssim
||p_G-p_{G_0'}||_1 + \varepsilon^2 + ||p_{G_0^\ast}-p_{G_0}||_1 \lesssim \varepsilon^2.
\end{split}\]
We can now invoke Lemma A.10 in Scricciolo (2011), p. 305, taking into account Remark A.3 of the same article. To this aim, note that, if $G$ satisfies \eqref{eq:condmixing}, then $G([-(a_\varepsilon+1),\,(a_\varepsilon+1)])>1/2$. The reader may also refer to Scricciolo (2014), p. 305. 
For any $G\in\mathscr G$, let $P_G$ stand for the probability measure with density $p_G\in\mathscr P$. The inclusion
\[\bigg\{P_G:\,\sum_{j=1}^N|G([y_j-\varepsilon^2,\,y_j+\varepsilon^2])-p_j|\leq \varepsilon^2\bigg
\}\subseteq B_{\textrm{KL}}\big (P_0;\,\varepsilon^2\log^2(1/\varepsilon)\big
)\]
holds. To apply Lemma A.2 of Ghosal and van der Vaart (2001), p. 1260, note that,
for every $y_j$, $1\leq j\leq N$, we have $\alpha([y_j-\varepsilon^2,\,y_j+\varepsilon^2])\gtrsim\varepsilon ^{b'}$ for some finite constant $b'>0$. Thus,
\[\log\Pi(B_{\textrm{KL}}(P_0;\,\varepsilon^2\log^2(1/\varepsilon)))
\gtrsim -N\log(1/ \varepsilon) \asymp -\varepsilon^{-2/3}\log(1/\varepsilon).\]
Taking $\varepsilon_n:=\varepsilon\log(1/\varepsilon)$, we have
$
\Pi(B_{\textrm{KL}}(P_0;\,\varepsilon_n^2)
)\gtrsim \exp{(-C_1 n\varepsilon_n^2)}
$ 
and condition (2) is satisfied.\\

\item \emph{Verification of condition} (3)\\[2pt]
Condition (3) requires that there exists a finite constant $B>0$ such that
$\Pi(||p_G||_\infty>B\mid \Data)=o_{\mathbf{P}}(1)$.
If $||f||_\infty<+\infty$, then
$||p_G||_\infty\leq ||f||_\infty<+\infty$ for all $G\in\mathscr G$, see Lemma \ref{lem:l2hel}. In particular, $||p_0||_\infty=||p_{G_0}||_\infty\leq ||f||_\infty<+\infty$.
Taking $B:=||f||_\infty$, we have
$$\forall\,n\in\mathbb{N},\,\,\,
\Pi(||p_G||_\infty>B\mid \Data)=0\quad P_0^n\mbox{-almost surely},$$
and condition (3) is satisfied. For the (standard) Laplace kernel density, $||f||_\infty=1/2$.
\end{itemize}
The proof is thus complete and assertion \eqref{eq:l2norm} follows.
\qed
\end{theopargself}


\section*{Appendix D: Inversion inequalities}\label{appendix:wasserstein}
\begin{theopargself}
In this section, we state a result relating, for every real number $p\geq1$, the $L^p$-Wasserstein distance  
between any pair of mixing distributions $G,\,G'\in\mathscr{G}$ to the $L^2$-distance between the corresponding mixed densities $p_G=G\ast f$ and $p_{G'}=G'\ast f$, with a kernel density $f$ that is ordinary smooth in the sense of condition \eqref{eq:ft} stated below. Lemma \ref{lem:2} extends Lemma 7 of Gao and van der Vaart (2016), pp. 621--622, beyond the case of compactly supported mixing distributions to mixing distributions with finite moment generating functions on some neighborhood of zero $(-s_0,\,s_0)$, with $0<s_0<1$. If, furthermore, the kernel density is bounded, $||f||_\infty<+\infty$, then the inversion inequality in \eqref{eq:wasserstein} below also holds for the Hellinger metric in virtue of the following known result, which is reported for the reader's convenience.

\begin{lemma}\label{lem:l2hel}
For a given kernel density $f$, let $p_G=G\ast f$, with $G\in\mathscr G$.
If $||f||_\infty<+\infty$, then 
$$\forall\,G\in\mathscr G,\,\,\, p_G(x)
\leq ||f||_\infty\quad\mbox{for all $x\in\mathbb{R}$,}$$
and
\begin{equation}\label{eq:hel^2}
\forall\,G,\,G'\in\mathscr G,\,\,\, ||p_G-p_{G'}||_2^2\leq 4||f||_\infty h^2(p_G,\,p_{G'}).
\end{equation}
\end{lemma}

\smallskip

We now state and prove an inequality translating the $L^2$-norm and the
Hellinger distance between mixed densities into any $L^p$-Wasserstein distance, $p\geq 1$, between the corresponding mixing distributions.

\begin{lemma}\label{lem:2}
Let $G$ and $G'$ be probability measures on some Borel-measurable space $(\mathscr{Y},\,\mathscr{B}(\mathscr{Y}))$, $\mathscr Y\subseteq\mathbb{R}$, such that
the associated moment generating functions
$M_G(s)$ and $
M_{G'}(s)$ are finite for all $|s|<s_0$, with $0<s_0<1$.
Let $f$ be a probability density function on $\mathbb{R}$, with Fourier transform $\hat f$ satisfying,
for some real number $\beta>0$, the condition
\begin{equation}\label{eq:ft}
\inf_{t\in\mathbb{R}}(1+|t|^\beta)|\hat f(t)|>0.
\end{equation}
Let $d$ stand for the $L^2$-distance between the mixed densities $p_G=G\ast f$ and $p_{G'}=
G'\ast f$. Then, for any real number
$p\geq 1$,
\begin{equation}\label{eq:wasserstein}
\hspace*{-0.3cm}
W_p(G,\,G')\lesssim d^{1/(p+\beta)}\pt{\log\frac{1}{d} }^{(p+1/2)/(p+\beta)}\quad
\mbox{ for }\, d=||p_G-p_{G'}||_2 \, \mbox{ small enough}.
\end{equation}
If, in addition, $||f||_\infty<+\infty$, then the upper bound in \eqref{eq:wasserstein} also holds for $d$ being the Hellinger distance,
$d=h(p_G,\,p_{G'})$.
\end{lemma}

\begin{proof}
For any real number $h>0$, by the triangle inequality, we have
\begin{equation}\label{eq:wass}
W^p_p(G,\,G')\leq W^p_p(G,\, G\ast\Phi_h) + W^p_p(G\ast \Phi_h,\,G'\ast \Phi_h) +W^p_p(G'\ast \Phi_h,\, G'),
\end{equation}
where $\Phi_h$ stands for a zero-mean Gaussian probability measure with variance $h^2$, whose density is denoted by $\phi_h(\cdot):=h^{-1}\phi(\cdot/h)$, for $\phi$ the density of a standard normal r.v. $W$. The first and third terms on the right-hand side of
\eqref{eq:wass} can be bounded above as follows. By standard arguments, see, for instance, the proof of Theorem 2 in Nguyen~(2013), pp. 389--391,
\begin{equation}\label{eq:max}
\max\{W_p^p(G,\, G\ast\Phi_h),\, W_p^p(G'\ast\Phi_h,\,G')\}\leq E[|hW|^p]\lesssim h^p
\end{equation}
because $E[|W|^p]<+\infty$ for every real number $p>0$, hence, \emph{a fortiori}, for every real $p\geq1$.
Concerning the second term on the right-hand side of
\eqref{eq:wass}, reasoning as in Lemma 7 of Gao and van der Vaart~(2016), pp. 621--622, for any real number $M>0$,
\[W_p^p
(G\ast\Phi_h,\,G'\ast\Phi_h)\lesssim
\pt{\int_{|x|\leq M}+\int_{|x|>M}}|x|^p|
(G-G')\ast\phi_h(x)|\,\d x=:T_1+T_2,\]
where, for every $0<h \leq 1$,
\begin{equation}\label{eq:t1}
T_1\lesssim M^{p+1/2}||(G-G')\ast\phi_h||_2 \lesssim M^{p+1/2} h^{-\beta} ||p_G-p_{G'}||_2
\end{equation}
because $\sup_{t\in\mathbb{R}}|\hat \phi(h
t)|/|\hat f(t)|\lesssim h^{-\beta}$ in virtue of assumption \eqref{eq:ft}. To see it, note that assumption \eqref{eq:ft} implies the existence of a finite constant $L_f>0$ such that $(1+|t|^\beta)|\hat f(t)|\geq L_f$ for all
$t\in\mathbb{R}$. Therefore, if $0<h \leq 1$, 
\[\sup_{t\in\mathbb{R}}\frac{|\hat \phi(h
t)|}{|\hat f(t)|}\leq \frac{1}{L_f}
\sup_{t\in\mathbb{R}}[(1+|ht|^\beta)|\hat \phi(h
t)|]\times \sup_{t\in\mathbb{R}}
\bigg(\frac{1+|t|^\beta}{1+|ht|^\beta}\bigg)\lesssim h^{-\beta}.
\]
If $||f||_\infty<+\infty$, then the $L^2$-distance between $p_G$ and $p_{G'}$ in \eqref{eq:t1} can be replaced with the Hellinger distance (see Lemma \ref{lem:l2hel}), so that
\[
 T_1\lesssim M^{p+1/2} h^{-\beta} h(p_G,\,p_{G'}).
\]
We now deal with the term $T_2$. We preliminarily derive an instrumental inequality. 
For every $x\in\mathbb{R}$ and real numbers $p,\,u>0$, 
\[
\frac{p}{u}
e^{u|x|/p}=\frac{p}{u}\sum_{j=0}^{+\infty}\frac{(u|x|/p)^j}{j!}\geq |x|,
\]
whence
\begin{equation}\label{eq:23}
|x|^p\leq (p/u)^p e^{u|x|}<(p/u)^p (e^{-ux}+e^{ux}).
\end{equation}
Now fix any number $0<u<s_0<1$. Applying the inequalities in \eqref{eq:23}
and taking into account the expression of the moment generating function of a standard Gaussian distribution $M_{\Phi}(s)=e^{s^2/2}$, $s\in\mathbb{R}$, we get 
\begin{align*}
\int_{-\infty}^{+\infty}\max\{1,\,|x|^p\}e^{u|x|}\phi_h(x)\,\d x&\leq 
\int_{-\infty}^{+\infty}\max\{e^{u|x|},\,(p/u)^p e^{2u|x|}\}\phi_h(x)\,\d x\\
&<2\max\{e^{(u h)^2/2},\,(p/u)^p e^{2(u h)^2}\}\\
&<
2\max\{e^{s_0^2/2},\,(p/u)^p e^{2s_0^2}\},
\end{align*}
namely, for fixed $u$, the above integral can be bounded above by a constant that is fixed throughout and can therefore be neglected when bounding $T_2$. Hence, 
\[\begin{split}
T_2 &\lesssim e^{-uM} \int_{|x|>M}|x|^pe^{u|x|}[(G+G')\ast\phi_h(x)]\,\d x\\
&\lesssim e^{-uM}\int_{\mathscr Y} (1+|y|^p)e^{u|y|}\pt{\int_{-\infty}^{+\infty}\max\{1,\,|x|^p\}e^{u|x|}\phi_h(x)\,\d x}\,\d (G+G')(y)\\
& \lesssim e^{-uM}\int_{\mathscr Y}(1+|y|^p)e^{u|y|}\,\d (G+G')(y) \lesssim e^{-uM}
\end{split}\]
because
\[\begin{split}
\int_{\mathscr Y}e^{u|y|}\,\d (G+G')(y) &<\int_{\mathscr Y}
(e^{-uy}+e^{uy})\,\d (G+G')(y)\\
&=(M_G+M_{G'})(-u)+(M_G+M_{G'})(u)<+\infty
\end{split}\]
and, for any fixed real number $0<\xi<1$ such that $0<s:=(\xi+u)<s_0$, by the inequalities in \eqref{eq:23},
\[\begin{split}
\int_{\mathscr Y}|y|^pe^{u|y|}\,\d (G+G')(y) &< 
(p/\xi)^p
\int_{\mathscr Y}e^{(\xi+u)|y|}\,\d (G+G')(y)\\&= 
(p/\xi)^p
\int_{\mathscr Y} e^{s|y|}\,\d (G+G')(y)\\&< (p/\xi)^p\int_{\mathscr Y}
(e^{-sy}+e^{sy})\,\d (G+G')(y)\\&=(p/\xi)^p[(M_G+M_{G'})(-s)+(M_G+M_{G'})(s)]<+\infty
\end{split}\]
by the assumption that both $G$ and $G'$ 
have finite moment generating functions on $(-s_0,\,s_0)$, for $0<s_0<1$. 
Thus,
\begin{equation}\label{eq:t2}
T_2\lesssim e^{-u M}.
\end{equation}
Combining partial results in \eqref{eq:max}, \eqref{eq:t1} and \eqref{eq:t2},
we get
\begin{equation}\label{eq:wasserstein2}
W^p_p(G,\,G')\lesssim h^p + M^{p+1/2}h^{-\beta}d+e^{-uM}
\end{equation}
and the conclusion follows by
minimizing the expression in \eqref{eq:wasserstein2} with respect to $h$ and $M$, which, for sufficiently small $d$, implies taking
$M=O(\log (1/d))$ and $h^{p+\beta}=O(d \log^{p+1/2} (1/d))$.
\qed
\end{proof}

\begin{remark}
The standard Laplace kernel density is bounded, with $||f||_\infty=1/2$, and  satisfies condition \eqref{eq:ft} for $\beta=2$.
\end{remark}
\end{theopargself}


%
\bibliographystyle{}
\bibliography{}
%
%

\end{document}